\documentclass[11pt]{article}
\usepackage{epsfig}
\usepackage{amssymb,amsmath,amsthm,url}
\usepackage{multirow}
\usepackage{booktabs}
\usepackage{setspace}
\usepackage{dsfont}
\usepackage{times,amsfonts,eucal}
\usepackage{algorithm}
\usepackage{algorithmic}
\usepackage{graphicx,ifthen}
\usepackage{wrapfig}
\usepackage{latexsym}
\usepackage{color}
\usepackage{mathrsfs}
\usepackage{bm}
\usepackage{bbm}
\usepackage{srctex}

\usepackage{bm}
\newtheorem{theorem}{Theorem}[section]
\newtheorem{lemma}{Lemma}[section]

\newtheorem{cor}{Corollary}[section]

\newtheorem{assumption}{Assumption}[section]
\numberwithin{equation}{section}
\numberwithin{theorem}{section}
\numberwithin{lemma}{section}
\numberwithin{pro}{section}
\numberwithin{cor}{section}
\numberwithin{definition}{section}
\numberwithin{cons}{section}
\numberwithin{rem}{section}
\numberwithin{exa}{section}
\numberwithin{table}{section}
\numberwithin{figure}{section}

\newcommand{\intt}{\int\hspace{-.2cm}\int}

\def\beq{\begin{equation}}
\def\eeq{\end{equation}}
\def\bals{\begin{align*}}
\def\eals{\end{align*}}

\def\bal{\begin{align}}
\def\eal{\end{align}}

\allowdisplaybreaks

\begin{document}

\title{Structural break analysis for \\ spectrum and trace of covariance operators\footnote{This research was partially supported by NSF grant DMS 1407530, and the NSERC Discovery and Accelerator grant
}}

\author{
Alexander Aue\footnote{Department of Statistics, University of California, Davis, CA 95616, USA, email: \tt{[aaue,osonmez]@ucdavis.edu}}
\and Gregory Rice\footnote{Department of Statistics and Actuarial Science, University of Waterloo, Waterloo, ON, Canada, email: \tt{grice@uwaterloo.ca}}
\and Ozan S\"onmez$^\dagger$
}
\date{\today}
\maketitle

\begin{abstract}
\setlength{\baselineskip}{1.8em}
This paper deals with analyzing structural breaks in the covariance operator of sequentially observed functional data. For this purpose, procedures are developed to segment an observed stretch of curves into periods for which second-order stationarity may be reasonably assumed. The proposed methods are based on measuring the fluctuations of  sample eigenvalues, either individually or jointly, and traces of the sample covariance operator computed from segments of the data. To implement the tests, new limit results are introduced that deal with the large-sample behavior of vector-valued processes built from partial sample eigenvalue estimates. These results in turn enable the calibration of the tests to a prescribed asymptotic level. A simulation study and an application to Australian annual minimum temperature curves confirm that the proposed methods work well in finite samples. The application suggests that the variation in annual minimum temperature underwent a structural break in the 1950s, after which typical fluctuations from the generally increasing trendstarted to be significantly smaller.

\end{abstract}

\setlength{\baselineskip}{1.8em}

\section{Introduction}
\label{sec:intro}

In functional data analysis, a natural way to measure the variability of a sample is through the covariance operator and its eigenvalues. This basic idea motivates perhaps the most widely used tool in the analysis of functional data: functional principal component analysis (FPCA). FPCA entails projecting functional observations into a lower dimensional space spanned by a few functional principal components computed as eigenfunctions of an empirical covariance operator. Typically a small number of projections account for a large percentage of sample variation, often measured as size of the corresponding eigenvalues of the empirical covariance operator relative to its trace. When functional data are obtained via randomized experiments, it is reasonable to assume that the covariance structure is homogeneous throughout the sample, and, in this case, the principal components and spectra computed from the sample covariance operator correspond to population quantities with well-known optimality properties for dimension reduction. The interested reader is referred to Ramsay and Silverman (2005), Ferraty and Vieu (2006), and Horv\'ath and Kokoszka (2013) for text book treatments of functional data analysis, and to 
Shang (2014) for a survey of FPCA.

Frequently, however, functional data are obtained not by simple random sampling, but rather sequentially. A common example is the generation of functional observations by parsing long, dense records of a continuous time phenomenon, such as historical temperature data, into a functional time series, such as annual temperature profiles; see Aue et al.\ (2018), Aue and van Delft (2018), and van Delft et al.\ (2018). Sequences of the same kind also arise from sequential observations of other complex functional phenomena, such as functional magnetic resonance imaging and DNA minicircle evolution; see Aston and Kirch (2012) and Panaretos et al.\ (2010), respectively. Consequently, functional data often naturally display time series characteristics.


With such sequences of functional observations it is also often evident that their variability is not stable throughout the entire sample, rather they exhibit several periods of distinct levels and fluctuations. Providing a mechanism to identify data segments for which variability can be assumed stable is useful for several reasons. First, FPCA based analyses using the entire sample might be misleading in the presence of inhomogeneity in the variability in that either (1) the basis computed from the sample covariance operator may not be estimating the optimal basis for dimension reduction, and/or (2) statistics used to determine how many principal components to use, often based on sample eigenvalue estimates, may not perform as expected. As a result too few or many principal components could be considered in subsequent analyses. Breaks in the variability, as measured by eigenvalues, might also be of independent interest since they may signal a relevant change to the system under study. An example is given by structural breaks in the variability of annual minimum temperature curves constructed from historical records in Australia. It is seen below that after the removal of an increasing trend curve, variability begins to decrease in the 1950s. Methods for identifying and pinpointing the nature of such structural breaks in functional data, and further giving statistical significance to such findings, have not been developed, to the best of our knowledge.

In this paper,  tests for the constancy of the largest $d$ eigenvalues and trace of the empirical covariance operator of a functional time series are proposed and studied. The tests are based on comparing maximally selected quadratic forms derived from partial sample estimates of the eigenvalues of the covariance operator to the quantiles of their limiting distribution under the hypothesis that the sample is taken from a weakly dependent functional time series. This asymptotic result follows from a weak invariance principle for the vector-valued process of partial sample eigenvalue estimates that might be of independent interest.

This work is inspired by, and builds upon, a number of recent contributions in both the probability and statistics literature. In the setting of separable Hilbert space-valued random variables, Mas (2002) and Mas and Menneteau (2003) showed via perturbation theory that the central limit theorem, law of large numbers, and law of iterated logarithm hold for the spectra if analogous results can be established for the operators themselves. Kokoszka and Reimherr (2013) showed, under conditions similar to those used here, that eigenfunctions of sample covariance operators of weakly dependent functional time series are asymptotically Gaussian. Beran and Liu (2016) established the central limit theorem for the eigenvalue and eigenfunction estimates in functional data models under both short- and long-memory error conditions. With the goal of performing structural break analysis with finite-dimensional time series,  Aue et al.\ (2009) established a weak invariance principle for the process partial sample estimates of the covariance matrix, which was applied do derive structural break tests for the second-order structure of a vector-valued time series. Their results were extended to include strong approximations for partial sample spectra and principal components in Kao et al. (2018), which may be viewed as a finite-dimensional counterpart to this paper. Horv\'ath and Rice (2018+) considered similar methods in the context of high-dimensional linear factor models.  

There are several recent papers on two-sample and analysis of variance problems for functional data relevant to the present work.  Most closely related is Jaru\v{s}kova (2013), who developed two-sample and structural break tests for the covariance operator of independent, identically distributed functional data based on principal component projections, and Zhang and Shao (2015), who considered a two-sample test for the covariance operator of dependent functional data based on self-normalized statistics derived from eigenvalue estimates. Beran et al.\ (2016) developed a two-sample test for the equivalence of eigenspaces of two-sample covariance operators, while Pigoli et al.\ (2014) considered various metrics and two-sample tests for covariance operators. Boente et al.\ (2010) developed multi-sample tests under a common principal component assumption. Finally Fremdt et al.\ (2012) and Panaretos et al.\ (2010) studied two-sample tests for the second-order structure of Gaussian functional data.

The rest of the paper is organized as follows. In Section \ref{main}, the basic problem is formalized, and assumptions and asymptotic results for partial sample eigenvalue estimates are detailed. Applications of these results to test for the constancy of the eigenvalues of the covariance operator are developed in Section~\ref{change-sec}. The findings of a Monte-Carlo simulation study are presented in Section \ref{sim}, while the outcomes of an application to annual minimum temperature curves are reported in Section~\ref{data}. All procedures are implemented in the R package {\tt fChange} (see S\"onmez et al., 2017), which may be downloaded from the CRAN website. Section \ref{sec:conc} concludes. Technical details and proofs are collected in Appendices~\ref{proofs} and \ref{var-est}. Below, the following notations are used. Write $L^2([0,1]^\ell)$ for the space of square-integrable, real-valued functions defined on $[0,1]^\ell$. Let $\|\cdot\|$ denote the standard norm on $L^2([0,1]^\ell)$ induced by the inner product $\langle\cdot,\cdot\rangle$, the dimension $\ell$ being clear from the input function. The notation $\int$ may be used in place of $\int_0^1$, and $(y_k)$ short for a sequence $(y_i\colon i\in \mathbb{Z})$ indexed by the integers $\mathbb{Z}$.

\section{Framework} \label{main}

Suppose that functional observations $X_1,\ldots,X_n$ are generated by the model
\begin{align}\label{model-1}
X_i(t)= \mu(t) + \varepsilon_i(t),\qquad t\in[0,1],\; i\in\mathbb{Z},
\end{align}
where $\mu$ denotes the common mean function of the $X_i$ and $(\varepsilon_i\colon i\in\mathbb{Z})$ a sequence of centered error functions treated as stochastic processes with sample paths in $L^2([0,1])$. 
In order to solidify concepts, assume that $\mathbb{E}[\|X_i\|^2] < \infty$, and let
\begin{align*}
C^{(i)}(t,t^\prime) = \mbox{Cov}(X_i(t),X_i(t^\prime)),
\qquad t,t^\prime\in[0,1],
\end{align*}
denote the covariance kernel of $X_i$. On $L^2([0,1])$, $C^{(i)}$ defines the symmetric and positive definite Hilbert--Schmidt integral operator $c^{(i)}$ given by
$$
c^{(i)}(f)(t) = \int C^{(i)}(t,t^\prime)f(t^\prime)dt^\prime, \qquad f\in L^2([0,1]),
$$
whose eigenfunctions $\varphi_j^{(i)}$ are commonly termed the principal components of the process $X_i$. The associated nonnegative, real, and ordered eigenvalues $\lambda_j^{(i)}$ define the ``variance explained" by successive principal components. Given a sample $X_1,\ldots,X_n$ following \eqref{model-1}, one often wishes to estimate these principal components and eigenvalues in order to perform dimension reduction. Under the assumption that the sequence $(X_i)$ is strictly stationary, which in light of \eqref{model-1} is equivalent with the strict stationarity of the errors $(\varepsilon_i)$, it follows that $C^{(i)}=C$ for all $i$, where $C(t,t^\prime) = \mbox{Cov}(X_0(t),X_0(t^\prime))$. Similarly, $\varphi_j^{(i)}=\varphi_j$ and $\lambda_j^{(i)}=\lambda_j$. These common principal components may be estimated using the sample covariance kernel
\begin{align*}
\hat{C}(t,t^\prime) = \frac{1}{n} \sum_{i=1}^{n} (X_i(t)- \bar{X}(t))(X_i(t^\prime) - \bar{X}(t^\prime)),
\qquad t,t^\prime\in[0,1],
\end{align*}
where $\bar{X}= n^{-1} \sum_{i=1}^n X_i$, which in turn yields estimates $\hat{\lambda}_j$ and $\hat{\varphi}_j$ as solutions to the equations
\begin{align}\label{eigen-nai}
\hat{\lambda}_j \hat{\varphi}_j(t) = \int  \hat{C}(t,t^\prime)\hat{\varphi}_j(t^\prime)dt^\prime,
\qquad t\in[0,1].
\end{align}
A potential issue with this approach arises as follows: if the errors in \eqref{model-1} are non-stationary, for instance if their covariance $C^{(i)}$ changes within the sample, then principal components and eigenvalues defined in \eqref{eigen-nai} may not lead to optimal dimension reduction and/or summaries of variability. Defining
$$
{\bf \Lambda}_d^{(i)} = (\lambda_1^{(i)},\ldots,\lambda_d^{(i)})^\top \in \mathbb{R}^d,
$$
with $^\top$ signifying transposition, the foregoing motivates to study the null hypothesis
$$
H_0\colon {\bf \Lambda}_d^{(1)} = \cdots = {\bf \Lambda}_d^{(n)}
$$
and the alternative
$$
H_A\colon {\bf \Lambda}_d^{(1)} = \cdots =  {\bf \Lambda}_d^{(k^*)} \ne  {\bf \Lambda}_d^{(k^*+1)}=\cdots  = {\bf \Lambda}_d^{(n)},
$$
where $k^*=\lfloor \tau n \rfloor$, with $\tau \in (0,1)$. The alternative hypothesis $H_A$ describes the situation in which there is a structural break in the $d$ largest eigenvalues taking place at the unknown break point $k^*$. In order to test $H_0$, consider partial sample estimates of $C$ given by
\begin{align}\label{c-par-def}
\hat{C}_x(t,t^\prime)
= \frac{1}{n} \sum_{i=1}^{\lfloor nx \rfloor} (X_i(t)- \bar{X}(t))(X_i(t^\prime) - \bar{X}(t^\prime)),
\qquad x;t,t^\prime\in[0,1].
\end{align}
The estimate $\hat{C}_x$ may be used to define a partial sample estimate of $c$ as
\begin{align}\label{c-opt-par-def}
\hat{c}_x(f)(t) = \int \hat{C}_x(t,s)f(t^\prime)dt^\prime,
\qquad x;t\in[0,1].
\end{align}
For $x\in[0,1]$, let $\hat{\lambda}_j(x)$ denote the ordered eigenvalues of $\hat{c}_x$ with corresponding orthonormal eigenfunctions $\hat{\varphi}_{j,x}$. Throughout, the following assumptions will be invoked regarding strict stationarity of the underlying functional time series, the level of serial dependence between successive functions in the sample, and the spacing of the population eigenvalues $(\lambda_j\colon j\in\mathbb{N})$.

\begin{assumption}\label{edep}
It is assumed that

(a) there is a measurable function $g\colon S^\infty\to L^2([0,1])$, where $S$ is a measurable space, and independent, identically distributed (iid) innovations $(\epsilon_i\colon i\in\mathbb{Z})$ taking values in $S$ such that $\varepsilon_i=g(\epsilon_i,\epsilon_{i-1},\ldots)$ for $i\in\mathbb{Z}$;

(b) there are $\ell$-dependent sequences $(\varepsilon_{i,\ell}\colon i\in\mathbb{Z})$ such that, for some $p>4$,
\[
\sum_{\ell=0}^\infty\big(\mathbb{E}[\|\varepsilon_i-\varepsilon_{i,\ell}\|^p]\big)^{1/p}<\infty,
\]
where $\varepsilon_{i,\ell}=g(\epsilon_i,\ldots,\epsilon_{i-\ell+1},\epsilon^*_{i,\ell,i-\ell},\epsilon^*_{i,\ell,i-\ell-1},\ldots)$ with $\epsilon^*_{i,\ell,j}$ being independent copies of $\epsilon_{i,0}$ independent of $(\epsilon_i\colon i\in\mathbb{Z})$.
\end{assumption}

Assumption \ref{edep}(a) implies that $(\varepsilon_i)$ is strictly stationary, and hence that $H_0$ holds. Processes satisfying Assumption \ref{edep}(b) were termed $L^p$-$m$-approximable processes by H\"ormann and Kokoszka (2010), and cover most stationary functional time series models of interest, including functional AR and ARMA processes (see Aue et al.\ 2015; and Bosq, 2000). It is assumed that the underlying error innovations $(\epsilon_i)$ are elements of an arbitrary measurable space $S$. However, in many examples $S$ is itself a function space, and the evaluation of $ g(\epsilon_{i},\epsilon_{i-1},\ldots)$ is a functional of $(\epsilon_j\colon j\leq i)$. In order to obtain a normal approximation for the sample eigenvalues of $\hat{c}$, one must assume at least $p=4$ moments for the norm of the observations, and so our assumption of $p>4$ is nearly optimal in this sense.

\begin{assumption}\label{d-lam}
There exists an integer $d\ge1$ such that $\lambda_1 > \cdots > \lambda_d > \lambda_{d+1} \ge 0$.
\end{assumption}

Assumption \ref{d-lam} is standard in the FPCA literature. It ensures that eigenspaces belonging to the $d$ largest eigenvalues of $c$ are one-dimensional and that $\min_{1 \le i \le d} (\lambda_i - \lambda_{i+1})$ is bounded away from zero. Under $H_0$, denote the vector of the $d$ largest eigenvalues of $c$ by
$$
{\bf \Lambda}_d = (\lambda_1,\ldots,\lambda_d)^\top \in \mathbb{R}^d.
$$
To consider tests based on the  vector of partial sample estimates of ${\bf \Lambda}_d$, define
$$
\hat{{\bf \Lambda}}_d(x) = ( \hat{\lambda}_1(x),\ldots,\hat{\lambda}_d(x))^\top,
\qquad x\in[\delta,1],
$$
and note that this gives rise to the process $(\hat{{\bf \Lambda}}_d(x)\colon x\in[\delta,1])$ living in $D^d[\delta,1]$, the $d$-dimensional Skorohod space on the interval $[\delta,1]$ with some $\delta\in (0,1)$; see Chapter 3 of Billingsley (1968). Let
\begin{align}\label{theta-def}
\theta_{i,j} = \langle\varepsilon_i\otimes\varepsilon_i-\mathbb{E}[\varepsilon_0\otimes\varepsilon_0],
\varphi_j\otimes\varphi_j\rangle,
\qquad i=1,\ldots,n;\; j=1,\ldots,d.
\end{align}
The following theorem establishes the asymptotic properties of a suitably normalized version of the process $(\hat{{\bf \Lambda}}_d(x)\colon x\in[\delta,1])$.

\begin{theorem}\label{thm-joint}
If model \eqref{model-1} and Assumptions \ref{edep} and \ref{d-lam} hold, then for $\delta\in(0,1)$,
$$
\sqrt{n}\left(\hat{{\bf \Lambda}}_d(x) - \frac{\lfloor n x \rfloor}{n} {\bf \Lambda}_d\colon x\in[\delta,1] \right) 
{\Longrightarrow} \big(\Sigma_d^{1/2}{\bf W}^{(d)}(x)\colon x\in[\delta,1]\big)
\qquad(n\to\infty),
$$
where $\Longrightarrow$ denotes weak convergence in $D^d[\delta,1]$, $({\bf W}^{(d)}(x)\colon x\in[0,1])$ a standard $d$-dimensional Brownian motion, and $\Sigma_d$ a $d\times d$ covariance matrix with entries
$$
\Sigma_d(j,j^\prime)=
\sum_{i=-\infty}^\infty \mathrm{Cov}(\theta_{0,j},\theta_{i,j^\prime}),
\qquad j,j^\prime=1,\ldots,d.
$$
\end{theorem}

For $i\in\mathbb{Z}$, let ${\bf \Theta}_i = (\theta_{i,1},\ldots,\theta_{i,d})^\top$, with $\theta_{i,j}$ as in \eqref{theta-def}. It is seen that $\Sigma_d$ is the usual long-run covariance matrix (or spectral density matrix at frequency zero) of the stationary sequence $({\bf \Theta}_i)$ in $\mathbb{R}^d$. Assuming for the moment that the series $(X_i)$ satisfying Assumption \ref{edep} and \ref{d-lam} was iid, $\Sigma_d(j,j)$ in Theorem \ref{thm-joint} would reduce to $2\lambda_j$, coinciding with standard asymptotic normality results for the eigenvalues computed from sample covariance operators based on a simple random sample; see Mas and Menneteau (2003).
%
As a corollary to Theorem \ref{thm-joint}, the limiting distribution of the individual partial sample empirical eigenvalue estimates is obtained. These asymptotics are useful in evaluating whether individual eigenvalues have undergone a structural break.

\begin{cor}\label{th-1}
{ If model \eqref{model-1} and Assumptions \ref{edep} and \ref{d-lam} hold, then for $j=1,\ldots,d$ and $\delta \in(0,1)$,
\begin{align}\label{th-1-eq1}
\sqrt{n} \left( \hat{\lambda}_j(x) -\frac{\lfloor nx \rfloor}{n} \lambda_j\colon x\in[\delta,1]\right) 
\Longrightarrow \big(\sigma_jW(x)\colon x\in[\delta,1]\big)
\qquad(n\to\infty),
\end{align}
where $\Longrightarrow$ denotes weak convergence in $D^d[\delta,1]$, $(W(x)\colon x\in[0,1])$ a standard one-dimensional Brownian motion, and
$
\sigma_j^2 =\Sigma_d(j,j).
$
}
\end{cor}

\section{Structural breaks in the covariance operator}
\label{change-sec}

\subsection{Testing for structural breaks in the spectrum}

As documented in Aue and Horv\'ath (2013) in univariate and multivariate contexts, a natural way to measure the validity of $H_0$ is to consider the magnitude of the vector-valued cumulative sum process
$$
\hat{{\bf \Lambda}}_d(x) - \frac{\lfloor n x \rfloor}{n}\hat{{\bf \Lambda}}_d(1),
$$
maximized over the partial sample parameter $x\in[\delta,1]$. Large values of this magnitude would be interpreted as evidence of inhomogeneity of the eigenvalues. Theorem \ref{thm-joint} may be used to determine the typical size of such a maximum. In order to pursue this goal the following assumption is imposed.

\begin{assumption}\label{mat-inv} The matrix $\Sigma_d$ defined in Theorem \ref{thm-joint} is invertible and there is an estimator $\hat{\Sigma}_d$ of $\Sigma_d$ satisfying
\begin{align}\label{sig-cons}
  |\Sigma_d - \hat{\Sigma}_d |_F = o_P(1),
\end{align}
where $|\cdot|_F$ is the Frobenius norm.
\end{assumption}
Appendix \ref{var-est} outlines a way to construct such a covariance estimator. There, a kernel lag-window type estimator
$$
\hat{\Sigma}_d = \sum_{\ell= -\infty}^{\infty} w\bigg(\frac{\ell}{h}\bigg) \hat{{\bf \Gamma}}_{\ell,\theta},
\qquad
\hat{{\bf \Gamma}}_{\ell,\theta}=\frac{1}{n}\sum_{i\in\mathcal{I}_\ell}\big(\hat{{\bf \Theta}}_i-\bar{{\bf \Theta}}\big)\big(\hat{{\bf \Theta}}_{i+\ell}-\bar{{\bf \Theta}}\big)^\top,
$$
of $\Sigma_d$ is discussed in some detail. Here, $w$ denotes a weight function and $h$ a bandwidth parameter, $\mathcal{I}_\ell=\{1,\ldots,n-\ell\}$ if $\ell\geq 0$ and $\mathcal{I}_\ell=\{1-\ell,\ldots,n\}$ if $\ell<0$, and $\hat{{\bf \Theta}}_j = (\hat{\theta}_{i,1},\ldots,\hat{\theta}_{i,d})^\top$ is the estimated score vector whose entries are given by
\begin{align}\label{theta-def-est}
 \hat{\theta}_{i,j} = \langle (X_i-\bar{X})\otimes (X_i -\bar{X})-\hat C_1,
 \hat\varphi_{j,1}\otimes\hat\varphi_{j,1}\rangle, 
\end{align}
while $\bar{\bf \Theta}$ is the sample mean of the $\hat{\bf \Theta}_i$. It is shown in Appendix~\ref{var-est} that this estimator satisfies \eqref{sig-cons} under standard conditions on the weight function $w$ and the bandwidth $h$. In order to test $H_0$, consider then the quadratic form statistic
$$
J_n(\delta)=J_{d,n}(\delta) = \sup_{\delta \le x \le 1} {\bf \kappa}_n^\top(x) \hat{\Sigma}_d^{-1} {\bf \kappa}_n(x),
$$
where
$$
{\bf \kappa}_n(x) = \sqrt{n}\left(\hat{{\bf \Lambda}}_d(x) - \frac{\lfloor n x \rfloor}{n}\hat{{\bf \Lambda}}_d(1)\right),
\qquad x\in[0,1].
$$
To evaluate the constancy of individual eigenvalues, consider the test statistic
$$
I_{j,n}(\delta)= \sup_{\delta \le x \le 1} \frac{1}{\hat{\sigma}_j} \left| \hat{\lambda}_j(x) - \frac{\lfloor nx \rfloor}{n} \hat{\lambda}_j(1)\right|,
\qquad j=1,\ldots,d,
$$
where $\hat{\sigma}^2_j = \hat{\Sigma}_d(j,j)$. The following result is a consequence of Theorem \ref{thm-joint}.

\begin{theorem}
If the conditions of Theorem \ref{thm-joint}, Assumption \ref{mat-inv} and \eqref{sig-cons} are satisfied, then
$$
J_n(\delta) \stackrel{D}{\to} J(\delta)=\sup_{\delta \le x \le 1} \sum_{j=1}^{d} B_j^2(x)
\qquad (n\to\infty),
$$
and
$$
I_{j,n}(\delta) \stackrel{D}{\to} I(\delta)=\sup_{\delta \le x \le 1} |B_j(x)|
\qquad(n\to\infty),
$$
where $\stackrel{D}{\to}$ indicates convergence in distribution and $(B_j(x)\colon x \in [0,1])$, $j=1,\ldots,d$, are iid standard Brownian bridges, noting that $I(\delta)$ does not depend on $j$.
\end{theorem}

A test of asymptotic size $\alpha$ for $H_0$ is to reject if $J_n(\delta)$ or $I_{j,n}(\delta)$ exceed the $1-\alpha$ quantile of the limit distributions distribution $J(\delta)$ and $I(\delta)$, respectively.
These distributions can be obtained via Monte-Carlo simulation. Below, the test based on $J_n(\delta)$ is referred to as the {\it joint test}, the test based on $I_{j,n}(\delta)$ as the {\it $j$th test} or the {\it $j$th individual test}.

\subsection{Testing for structural breaks in the trace}

The eigenvalue $\lambda_j$ is used to determine the variance of $X_0$ explained by the $j$th principal component $\varphi_j$ by comparing its magnitude to the {\it cummulative variance} of the function $X_0$ measured by the trace of the covariance operator
$$
\sum_{j=1}^{\infty}\lambda_j = \int C(t,t)dt=\mathrm{tr}(c).
$$
A common criterion for selecting the number of principal components for subsequent analysis is to take the minimum $d$ that causes the {\it total variance explained (TVE)} by the first $d$ principal components to exceed a user selected threshold $v$, that is,
\begin{equation}\label{TVE}
d=d_v = \min \left\{ d\colon \frac{\lambda_1 + \cdots + \lambda_d}{\mathrm{tr}(c)} \ge v   \right\}.
\end{equation}
When performing principal component analysis for functional time series it is often also of interest to determine if $\mathrm{tr}(c)$ is constant in conjunction with the constance of the largest eigenvalues. A partial sample estimator of the trace is given by
\begin{align}\label{trace-def}
T_n(x) = \frac{1}{n} \sum_{i=1}^{\lfloor nx \rfloor} \|X_i-\bar X\|^2,
\qquad x\in[0,1].
\end{align}
The large-sample behavior of a centered version of the process $(T_n(x)\colon x\in[0,1])$ is given next.

\begin{theorem}\label{thm-trace}
If Assumptions \ref{edep} and \ref{d-lam} hold, then
$$
\big(T_n(x) - x\,\mathrm{tr}(c)\colon x\in[0,1]\big)
\Longrightarrow 
(\sigma_T W(x)\colon x\in[0,1])
\qquad(n\to\infty),
$$
where $\Longrightarrow$ denotes weak convergence in $D[0,1]$, $(W(x)\colon x\in[0,1])$ a standard Brownian motion and, with $\xi_i=\|X_i-\mu\|^2$,
\begin{align}\label{sig-trace-def}
\sigma_T^2 = \sum_{i=-\infty}^{\infty} \mathrm{Cov}(\xi_0,\xi_i).
\end{align}
\end{theorem}
\noindent Utilizing Theorem~\ref{thm-trace} to test for a structural break in the trace of the covariance operator, one may set up the test statistic
$$
M_n = \frac{1}{\hat{\sigma}_T} \sup_{0\le x \le 1} |T_n(x) - x T_n(1)|,
$$
with a consistent estimator of $\sigma_T^2$ of the form
$$
\hat{\sigma}_T^2 = \sum_{\ell=-\infty}^{\infty} w \left( \frac{\ell}{h} \right) \hat{\gamma}_\ell,
\qquad
\hat{\gamma}_\ell = \frac{1}{n} \sum_{i\in \mathcal{I}_\ell} (\hat{\xi}_i - \bar{\xi})(\hat{\xi}_{i+\ell} - \bar{\xi}),
$$
where $w$ is a weight function, $h$ a bandwidth parameter, $\hat\xi_i=\|X_i-\bar X\|^2$ and $\mathcal{I}_\ell$ is as above. The consistency of this estimator under standard assumptions on $w$ and $h$ is discussed in Appendix \ref{var-est}. The following result is a consequence of Theorem~\ref{thm-trace}.

\begin{cor}
If the conditions of Theorem~\ref{thm-trace} are satisfied and if $\hat\sigma_T^2$ is consistent for $\sigma_T^2$, then
$$
M_n \stackrel{D}{\to} M= \sup_{0 \le x \le 1} |B(x)|
\qquad (n\to \infty),
$$
where $(B(x)\colon x \in [0,1])$ is a standard Brownian bridge.
\end{cor}
As for the joint and the individual tests above, a test of asymptotic size $\alpha$ for the null of no structural break in the trace is to reject if $M_n$ exceeds the $1-\alpha$ quantile of the limit distribution $M$. This test will be referred to as the {\it trace test} below.

\subsection{Consistency of test statistics}

In this subsection, the test statistics proposed above are shown to be consistent under $H_A$. To this end, suppose that the functional time series is stationary and weakly dependent before and after the break point $k^*$, and that an additional regularity condition is satisfied to ensure that the matrix estimate $\hat{\Sigma}_d$ does not have eigenvalues diverging to $+\infty$ under $H_A$. All details are specified in the following assumption.

\begin{assumption}\label{alt-as}

(a) There are measurable functions $g_1,g_2\colon S^\infty\to L^2([0,1])$, where $S$ is a measurable space, and iid innovations $(\epsilon_i)$ taking values in $S$ such that
$$
\varepsilon_i=\left\{
\begin{array}{l@{\quad}l}
      g_1(\epsilon_i,\epsilon_{i-1},\ldots), &  i \le k^*, \\[.2cm]
      g_2(\epsilon_i,\epsilon_{i-1},\ldots), &  i > k^*,
\end{array} \right.
$$
for $i\in\mathbb{Z}$, where $g_1$ and $g_2$ satisfy Assumption \ref{edep}(b). Let $C_1(t,t^\prime) = \mathrm{Cov}(\varepsilon_1(t),\varepsilon_1(t^\prime))$ and $C_2(t,t^\prime) =\mathrm{Cov}(\varepsilon_n(t),\varepsilon_n(t^\prime))$. Let $(\lambda_j^{(1)},\lambda_j^{(2)})$ and $(\varphi_j^{(1)},\varphi_j^{(2)})$ denote the eigenelements of $C_1$ and $C_2$, respectively. 

(b) Let $(\lambda_j^*,\varphi_j^*)$ denote the eigenelements of the integral operator $c^*$ with kernel
$$
C^*(t,t^\prime)= \tau C_1(t,t^\prime) + (1-\tau)C_2(t,t^\prime),
$$
where $\tau\in(0,1)$ is such that $k^*=\lfloor\tau n\rfloor$. Assume then invertibility of the matrices $\Sigma_d^{(1)}$ and $\Sigma_d^{(2)}$, whose entries are defined by
$$
\Sigma_d^{(k)}(j,j^\prime)= \sum_{i=-\infty}^\infty \mathrm{Cov}(\theta_{0,j}^{(k)},\theta_{i,j^\prime}^{(k)}),
\qquad j,j^\prime=1,\ldots,d;\; k=1,2,
$$
with
\begin{align*}
\theta^{(k)}_{i,j} = \langle\varepsilon_i\otimes\varepsilon_i-\mathbb{E}[\varepsilon_i\otimes\varepsilon_i],
\varphi_j^*\otimes\varphi_j^*\rangle,
\qquad
\varepsilon_i= g_k(\epsilon_i,\epsilon_{i-1},\ldots),~~k=1,2.
\end{align*}
\end{assumption}

Under \eqref{model-1}, Assumption \ref{alt-as} guarantees that the sequence $(X_i)$ is stationary and weakly dependent on the pre- and post-break segments. It is assumed that the first $d$ eigenspaces associated with pre- and post-break covariance operators are the same and one-dimensional. One notable feature of Assumption \ref{alt-as} is that the eigenfunctions of pre- and post-break covariance operators need not necessarily align. In particular, all proposed tests are expected to be consistent if both eigenvalues and eigenfunctions undergo a structural break, so long as $H_A$ holds.

\begin{theorem}\label{thm-cons}
Let Assumption \ref{alt-as} and $H_A$ be satisfied.

(a) If $\delta < \tau$, then
$J_n(\delta) \stackrel{P}{\to} \infty$ as $n\to \infty$;

(b) If $\delta<\tau$ and $\lambda_j^{(1)} \ne \lambda_j^{(2)}$, then
$I_{j,n}(\delta) \stackrel{P}{\to} \infty$ for $j=1,\ldots,d$ as $n\to \infty$;

(c) If $\int C_1(t,t) dt \ne \int C_2(t,t)dt$, then
$M_n \stackrel{P}{\to} \infty$ as $n\to \infty$.
\end{theorem}

The proof of Theorem \ref{thm-cons} is provided in Appendix \ref{var-est}. The main difficulty in establishing the result is deriving the asymptotic behaviour of $\hat{\Sigma}_d$ and the eigenvalues of its inverse under $H_A$.

\section{Simulation Study}
\label{sim}

\subsection{Setting}
\label{setting}

Data generating processes were considered following the setting of Aue et al.\ (2015, 2018). Specifically, functional data of sample size $n$ were generated utilizing $D=21$ Fourier basis functions $v_1, \dots, v_D$ on the unit interval $[0,1]$. 
The results reported below remained largely invariant to the choice of larger values of $D$.  Without loss of generality, the mean function $\mu$ in model (\ref{model-1}) was assumed to be the zero function. Independent curves were then constructed according to
$$
\zeta_i= \sum_{\ell=1}^D \xi_{i,\ell}v_\ell,\qquad i=1,\dots,n,
$$
where $\bm{\xi}_i = (\xi_{i,1}, \dots, \xi_{i,D})$ are independent normal random variables with zero mean and standard deviations $\bm{\sigma}=(\sigma_1, \dots, \sigma_D)$. Two standard deviations were chosen to mimic two different eigenvalue decays of the covariance operators, namely:
\begin{itemize}\itemsep-.3ex
\item[(a)] {\it Fast decay}: $\bm{\sigma}=(3^{-\ell}:\ell=1,\dots,D)$;
\item[(b)] {\it Slow decay}: $\bm{\sigma}=(\ell^{-1}:\ell=1,\dots,D)$.
\end{itemize}
To explore the finite-sample performance of the proposed tests, artificial breaks were inserted into the eigenvalue structures in (a) and (b) in the following way. For a fixed break location $k^* \in\{1,\dots,n\}$, consider
\[
\bm{\sigma}^{(1)} = \bm{\sigma}
\qquad\mbox{and}\qquad
\bm{\sigma}^{(2)} = \bm{b} \circ \bm{\sigma},
\]
where $\bm{\sigma}$ is as above, $\bm{b}=(b_1,\ldots,b_D)$ is a vector of sensitivity parameters and $\circ$ denotes the Hadamard product (entry-wise multiplication). Then, $\bm{\sigma}^{(1)}$ and $\bm{\sigma}^{(2)}$ specifiy the eigenvalue structure of the pre-and post-break observations with $\bm{b}$ controlling the magnitude of the break in a multiplicative fashion. For example, setting $\bm{b}=(1,\dots,1)$ results in the null hypothesis of structural stability, while $\bm{b} = (2,1,\dots,1)$ restricts the break to occur only in the leading eigenvalue, with $b_1=2$ determining the break size.

Both independent curves $\varepsilon_i=\zeta_i$ and functional time series curves were used, the latter to explore the effect of temporal dependence on the proposed tests. In particular, first-order FARs $\varepsilon_i = \Psi^{(k)}\varepsilon_{i-1}+\zeta_i$, $i=1,\dots,n$ and $k=1,2$, were generated (using a burn-in period of $n/2$ initial curves that were discarded). The operator was set up as $\Psi^{(k)} = \kappa\Psi_0^{(k)}$, where the random operator $\Psi_0^{(k)}$ was represented by a $D \times D$ matrix whose entries consisted of independent, centered normal random variables with standard deviations given by $\bm{\sigma}^{(k)}\bm{\sigma}^{(k)\top}$. A scaling was applied to achieve $||\Psi_0^{(k)}||=1$. The constant $\kappa$ can then be used to adjust the strength of temporal dependence. To ensure stationarity of the time series, $|\kappa|=0.8$ was selected.

With the above in place, the following four settings were studied.
\begin{itemize}\itemsep-.3ex
\item {\it Setting 1:} $b_1$ varies between $1, 1.5, 2, 3, 5$ and $b_\ell=1$ for all $\ell\neq1$;
\item {\it Setting 2:} $b_2$ varies between $1, 1.5, 2, 3, 5$ and $b_\ell=1$ for all $\ell\neq2$;
\item {\it Setting 3:} $b_3$ varies between $1, 1.5, 2, 3, 5$ and $b_\ell=1$ for all $\ell\neq3$;
\item {\it Setting 4:} $b_1$, $b_2$ and $b_3$ vary between $1, 1.25, 1.5, 1.75, 2$ and $b_\ell=1$ for all $\ell\neq 1,2,3$.
\end{itemize}

Settings 1--3 correspond to a structural break individually affecting the first, second and third eigendirections, respectively. Setting 4 allows for the leading three eigendirections to jointly undergo a structural break. All settings include the null hypothesis by setting all $b_\ell$ to unity. 

Combining the previous paragraphs, the functional curves $X_i = \varepsilon_i$, $i=1,\dots, n$, were generated according to model \eqref{model-1}. Simulations were run for both independent and FAR(1) curves for sample sizes $n=100$, $200$ and $500$ across the different specifications above and break locations $k^*=\lfloor\tau n\rfloor$ with $\tau=0.25,0.5$. For each data generating process, the individual test statistic $I_{j,n}(\delta)$, the joint test statistic $J_n(\delta)$ and the trace test statistic $M_n$ were applied to detect structural breaks, with $\delta=0.1$. All results reported in the next sections are based on 1{,}000 runs of the simulation experiments.

\subsection{Level and power of the detection procedures}
\label{power}

Empirical level and power of the proposed methods were evaluated relative to the nominal level $\alpha=0.05$. The results are presented in Table \ref{tab:level}. It can be seen that even for these rather small-to-moderate sample sizes, tests kept levels rather well across all specifications.

\begin{table}[h]
\vspace{.5cm}
\centering
\begin{tabular}{c@{\qquad}c@{\qquad}c@{\qquad}c@{\qquad}c@{\qquad}c@{\qquad}c@{\qquad}c}
\hline
Decay & DGP &$n$ & $J_n(\delta)$ & $I_{1,n}(\delta)$ & $I_{2,n}(\delta)$ & $I_{3,n}(\delta)$ & $M_n$  \\ \hline
slow & IID & \phantom{1}100 & 0.02 & 0.04 & 0.04 & 0.04 & 0.04 \\
&& 200 &  0.03 & 0.04 & 0.05 & 0.05 & 0.04 \\
&& 500 &  0.04 & 0.05 & 0.05 & 0.05 & 0.05 \\[.1cm]
& FAR(1) & \phantom{1}100 & 0.08 & 0.08 & 0.05 & 0.05 & 0.03  \\
&&200 & 0.06 & 0.07 & 0.03 & 0.03 & 0.04 \\
&& 500 &  0.06 & 0.06 & 0.06 & 0.06 & 0.05 \\
\hline
fast & IID & \phantom{1}100 & 0.03 & 0.04 & 0.04 & 0.04 & 0.04 \\
&&200 & 0.04 & 0.06 & 0.04 & 0.04 & 0.04\\
&& 500 &  0.04 & 0.05 & 0.05 & 0.05 & 0.05 \\[.1cm]
& FAR(1) & \phantom{1}100 & 0.07 & 0.07 & 0.04 & 0.04 & 0.03 \\
&&200 & 0.06 & 0.07 & 0.05 & 0.05 & 0.03  \\
&& 500 &  0.06 & 0.06 & 0.04 & 0.04 & 0.04 \\
\hline
\end{tabular}
\caption{Empirical sizes for the various detection procedures for two data generation processes. The nominal level was $\alpha=0.05$. $J_n(\delta)$ refers to the joint test for the first three eigenvalues, $I_{j,n}$, to the test for the $j$th eigenvalue, $j=1,2,3$, and $M_n$ to the the trace test; $\delta=0.1$.}
\label{tab:level}
\end{table}

To examine the power of the tests, structural breaks were inserted as described in Section \ref{setting}. The empirical rejection rates for each test statistic are reported as power curves in Figure~\ref{powerSlowIID} when the errors in model \ref{model-1} are iid curves, and the decay of the eigenvalues of the covariance operator is slow, as specified in setting (b) in the previous section.
Further simulation evidence is provided in the Appendix. The findings may be summarized as follows.

\begin{figure}[h]
\vspace{0.4cm}
\begin{center}
\includegraphics[width=\textwidth]{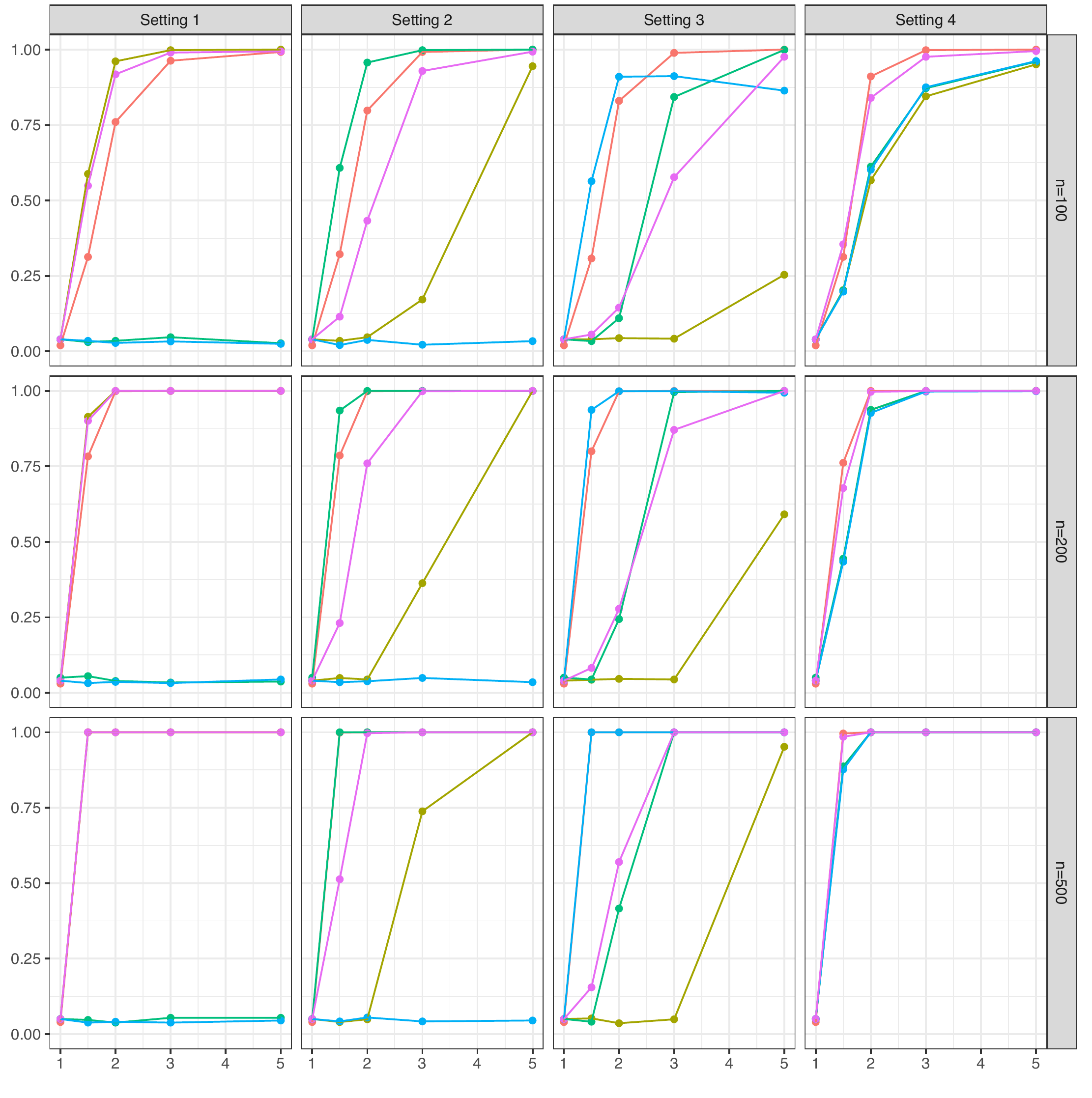}
\end{center}
\vspace{-0.4cm}
\caption{Power curves for the various test procedures for structural breaks in slowly decaying eigenvalues as in (b), iid functional curves and three sample sizes.}
\label{powerSlowIID}
\end{figure}

\begin{itemize}\itemsep-.4ex
\item When the break is dominant in a single eigenvalue, the corresponding individual eigenvalue test $I_{j,n}(\delta)$ tended to have reasonably high empirical power. The joint test $J_n(\delta)$ was generally competitive with its individual counterparts, losing some power due to the estimation of eigenvalues not contributing to the structural break.
\item Some care is required in the labeling of test statistics and settings. For instance, in the case that a sufficiently large break $b_2$ is inserted into the ``second'' eigendirection, this break will become dominant and constitute the leading mode of variation of the operator $c^*$ introduced in Assumption~\ref{alt-as}\,(b). It will therefore be picked up by the first individual test $I_{1,n}(\delta)$. This effect is most clearly seen in Figure \ref{powerSlowIID} for $b_3\geq 3$ and the test $I_{2,n}(\delta)$ predominantly picking up this break.
\item When the break is not dominant but spread out across the three largest eigenvalues as prescribed in Setting 4, then the advantage of the joint test $J_n(\delta)$ becomes more visible, especially for small sample sizes.
\item The test for breaks in the trace displays higher empirical power when the break occurs in larger eigenvalues, since these contribute more to total variation. Once the break is inserted in smaller eigenvalues, the trace test loses some power. As expected, this phenomenon is even more evident when the eigenvalues of the covariance operator have a fast decay (results not shown here).
\item The expected improvement in empirical power when $n$ increased was noted.
\end{itemize}

\subsection{Performance of break date estimates}
\label{estimate}

Once the null hypothesis of structural stability is rejected, it should be followed by an estimation of the break date. Assuming that model \eqref{model-1} and Assumptions \ref{edep} and \ref{d-lam} hold, the break date estimator $\hat{k}_{j,n}^* = \lfloor n\hat{x}_{j,n}^*\rfloor$ accompanying the $j$th individual test can be specified through
$$
\hat{x}_{j,n}^* = \arg\max_{\delta\leq x\leq1}\frac{1}{\hat{\sigma}_j} \bigg| \hat{\lambda}_j(x) - \frac{\lfloor nx \rfloor}{n} \hat{\lambda}_j(1)\bigg|,
$$
where $\delta\in(0,1)$ and $\hat{\sigma}^2_j$ is a consistent estimator of $\sigma_j^2$ as defined in Theorem \ref{th-1}. The break date estimator $\tilde{k}_n^* = \lfloor n\tilde{x}_n^*\rfloor$ accompanying the joint test can be set up with
$$
\tilde{x}_n^* = \arg\max_{\delta \le x \le 1} {\bf \kappa}_n(x)^\top \hat{\Sigma}_d^{-1} {\bf \kappa}_n(x),
$$
where $\kappa_n(x)$ and $\hat{\Sigma}_d$ are defined in Section 2.2. Finally, in a similar fashion, the break date estimator $\bar{k}_n^* = \lfloor n\bar{x}_n^*\rfloor$ for total variation is utilized with
$$
\bar{x}_n^* = \arg\max_{0\le x \le 1} \frac{1}{\hat{\sigma}_T}|T_n(x) - x T_n(1)|,
$$
where $T_n(x)$ is given in \eqref{trace-def}.

Settings 1--4 were used to insert eigenvalue breaks with scaling $b$ chosen to be $1.5$ and $3$. 
The slow decay of the eigenvalues in (b) above was considered. For each setting, sample size and choice of $b$, the break date estimators were applied to joint, first, second and third eigenvalue, and the trace tests. The results are presented in the form of boxplots in Figure \ref{powerSlowIID}. Overall the performance of the joint eigenvalue break date procedure is competitive with its marginal counterparts across all settings. However, the performance of the single eigenvalue break estimation procedures critically depends on the location and the magnitude of the break.

\begin{figure}[h!]
\vspace{-.4cm}
\begin{center}
\includegraphics[width=1\textwidth]{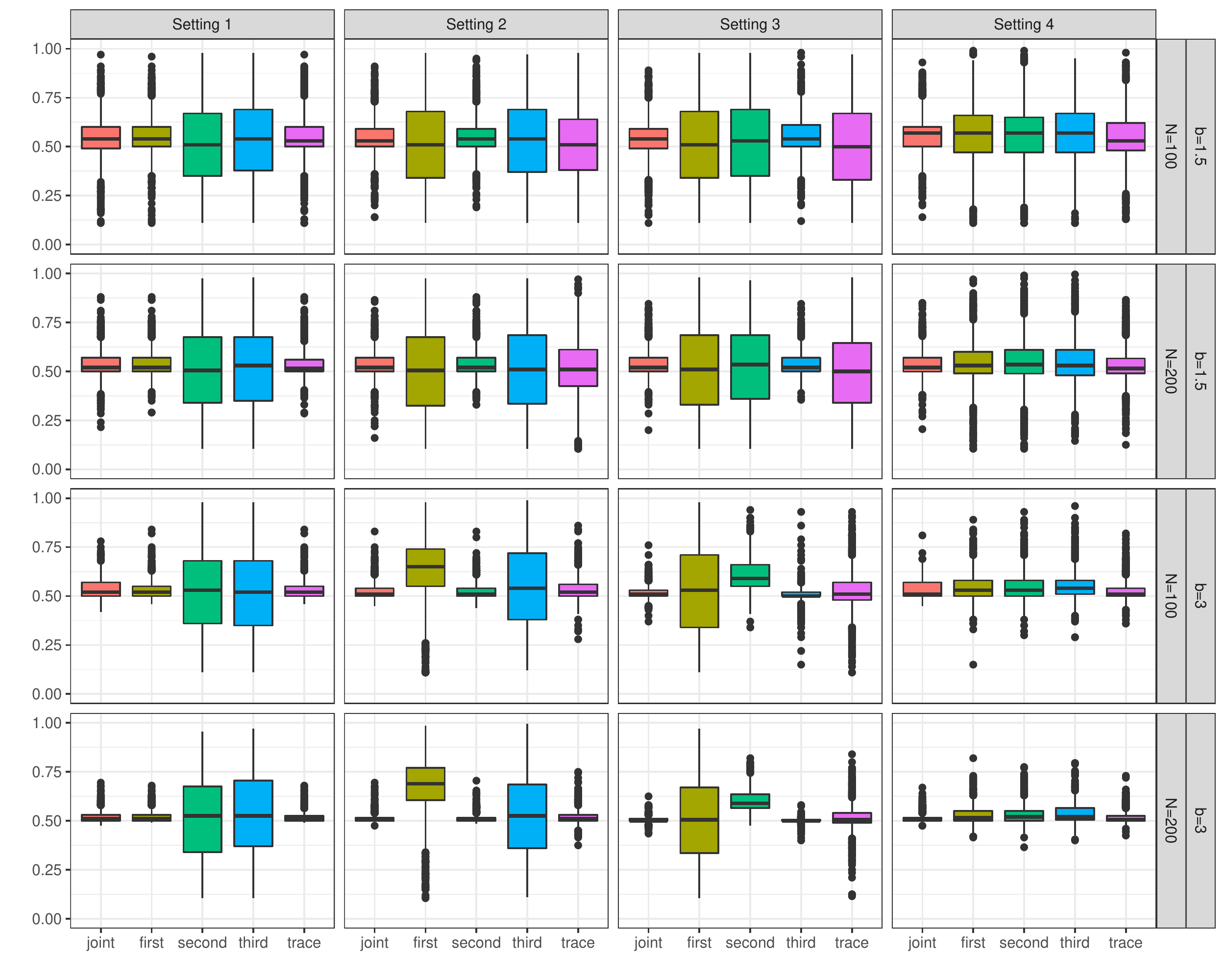}
\end{center}
\vspace{-0.4cm}
\caption{Boxplots for the various break date estimators for structural breaks in slowly decaying eigenvalues as in (b), iid functional curves, three sample sizes and two break magnitudes indexed by $b$.}
\label{powerSlowIID}
\end{figure}

\section{Application to annual temperature profiles}
\label{data}

This section is devoted to demonstrating the practical relevance of the proposed methods using annual temperature curves from various measuring stations in Australia. The raw data consists of daily minimum temperature measurements recorded in degrees Celsius over about one hundred years. For each year, 365 (366 in leap years) raw data points were converted into functional data object using $D=21$ Fourier basis functions. The data is available online and can be downloaded from \url{www.bom.gov.au}.
%
%
Here, attention is focused on a measuring station located at the Post Office in Gayndah, a small town in Queensland. For this particular station, full annual temperature profiles were available from  1894 until 2007, resulting into the $n = 115$ annual curves displayed in Figure \ref{Dataplots}.

Before attempting any structural break analysis for the covariance operator, the effect of potential nonstationarities in the mean function has to be taken into account. This can be done in several ways. Two approaches were discussed in S\"onmez (2018), namely binary segmentation based on the method of Aue et al.\ (2018) and moving average smoothing. Since both methods led to almost identical conclusions in terms of the structural break analysis for the mean curve, thus indicating some robustness with respect to the method of detrending, only results for binary segmentation are reported here. Its application yielded three data segments for which the mean function is reasonably constant. The corresponding breaks were located at $\hat k_1^*=60$ (1953) and $\hat k_2^*=79$ (1972). Plots corroborating the findings are given in Figure~\ref{Dataplots}, which indicate that the minimum temperature curves exhibit a generally increasing trend over the observed period.

\begin{figure}[h]
\vspace{-.4cm}
\begin{center}
\includegraphics[width=.41\textwidth]{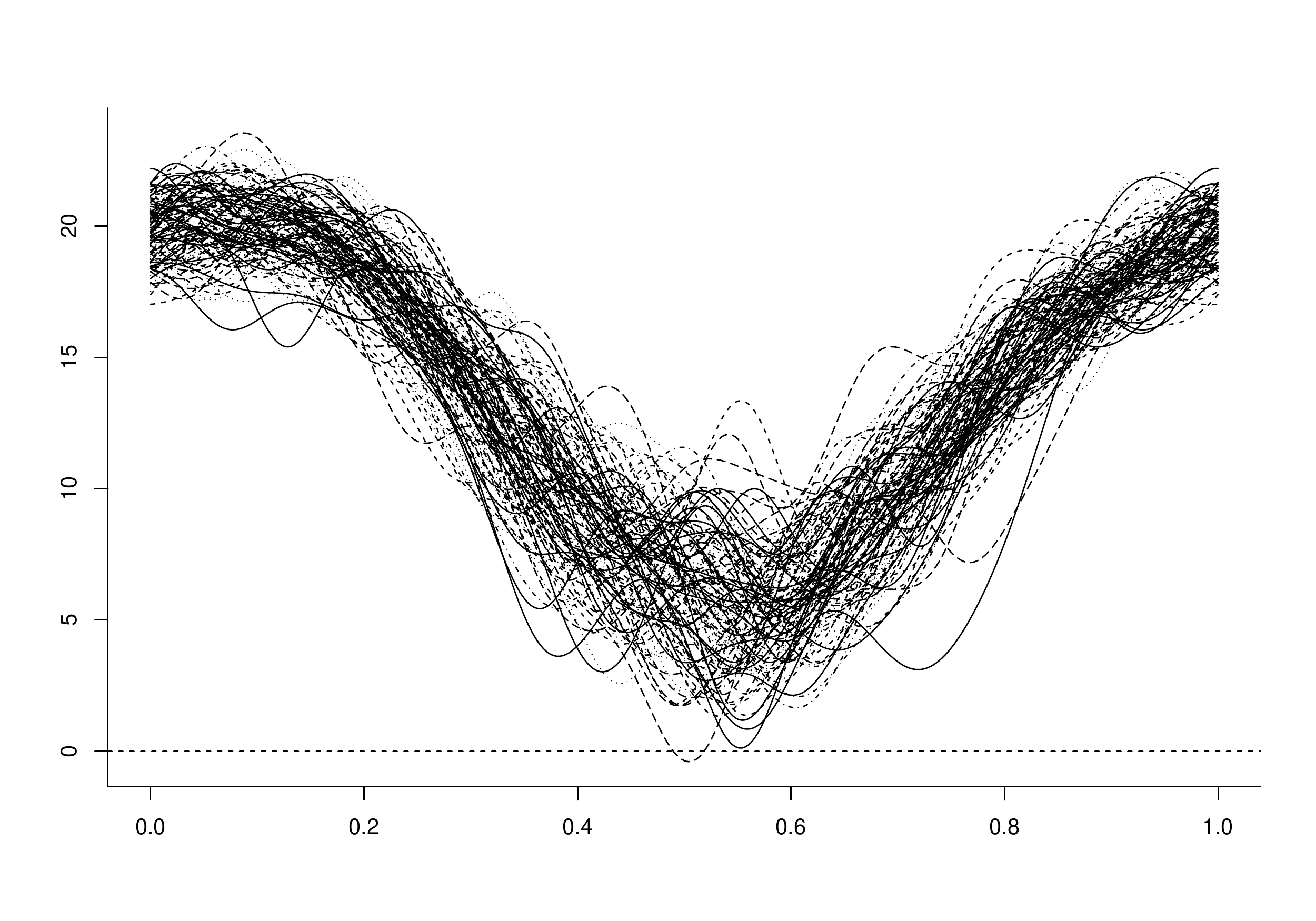}
\includegraphics[width=.41\textwidth]{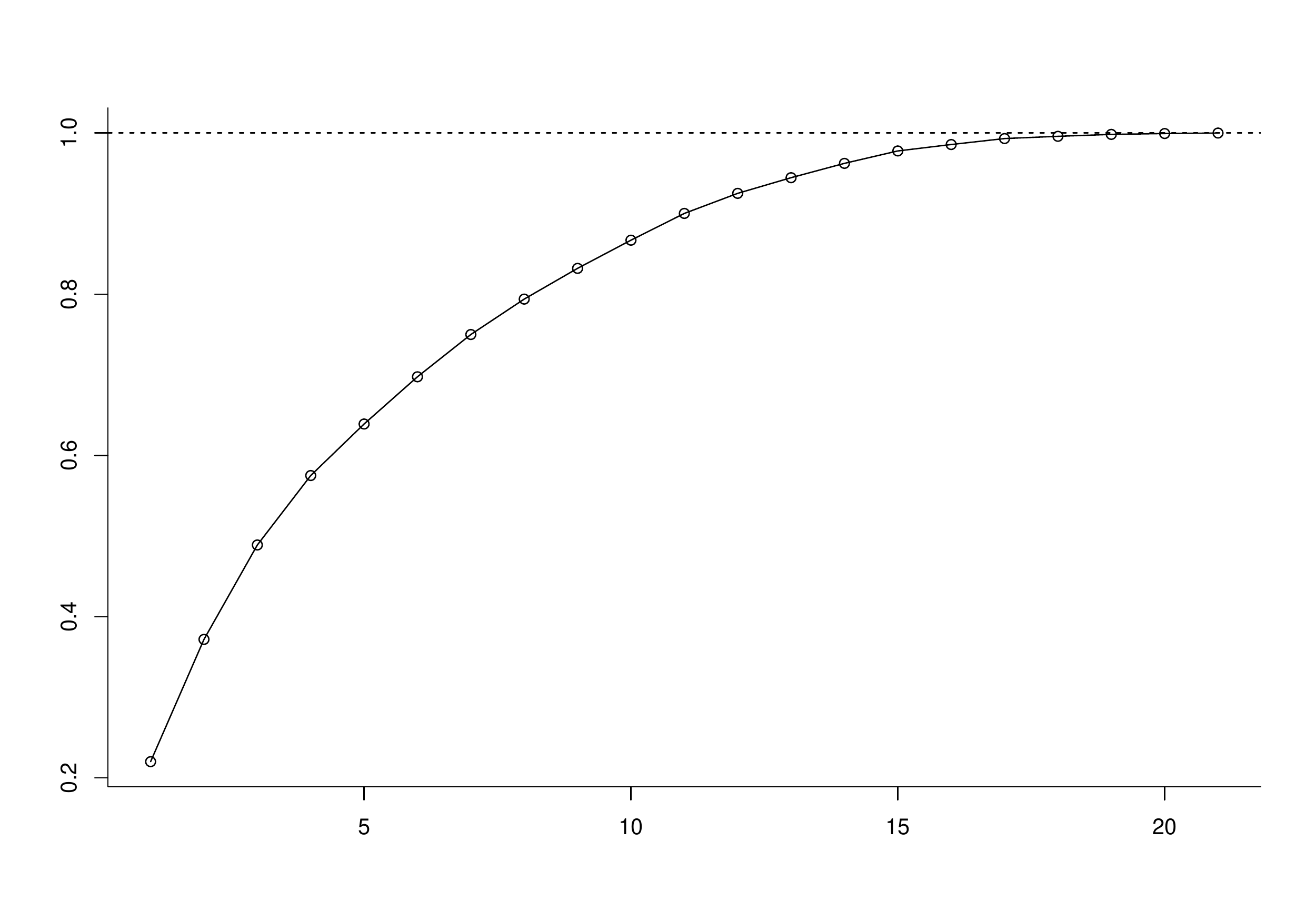} \\[-.3cm]
\includegraphics[width=.41\textwidth]{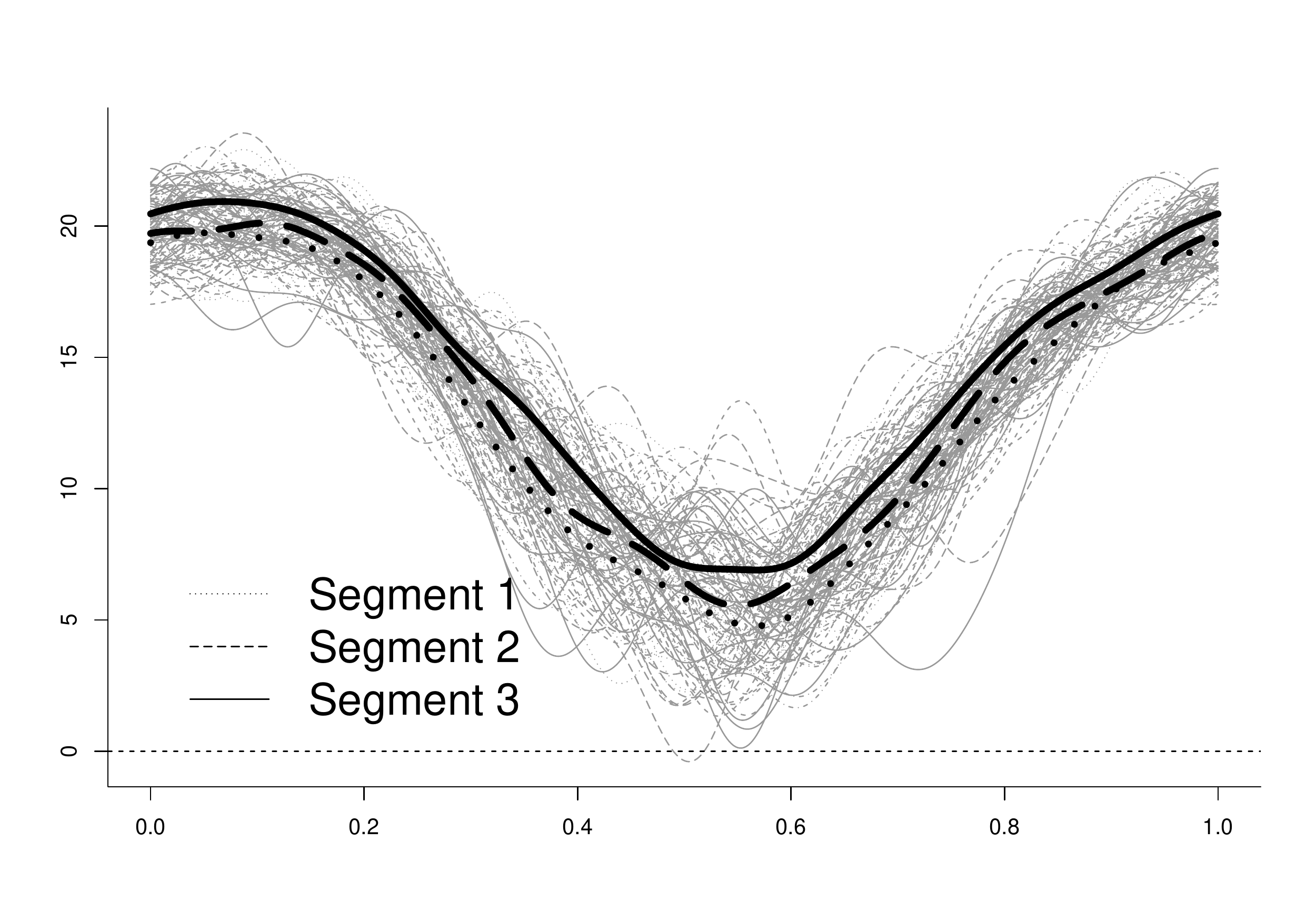}
\includegraphics[width=.41\textwidth]{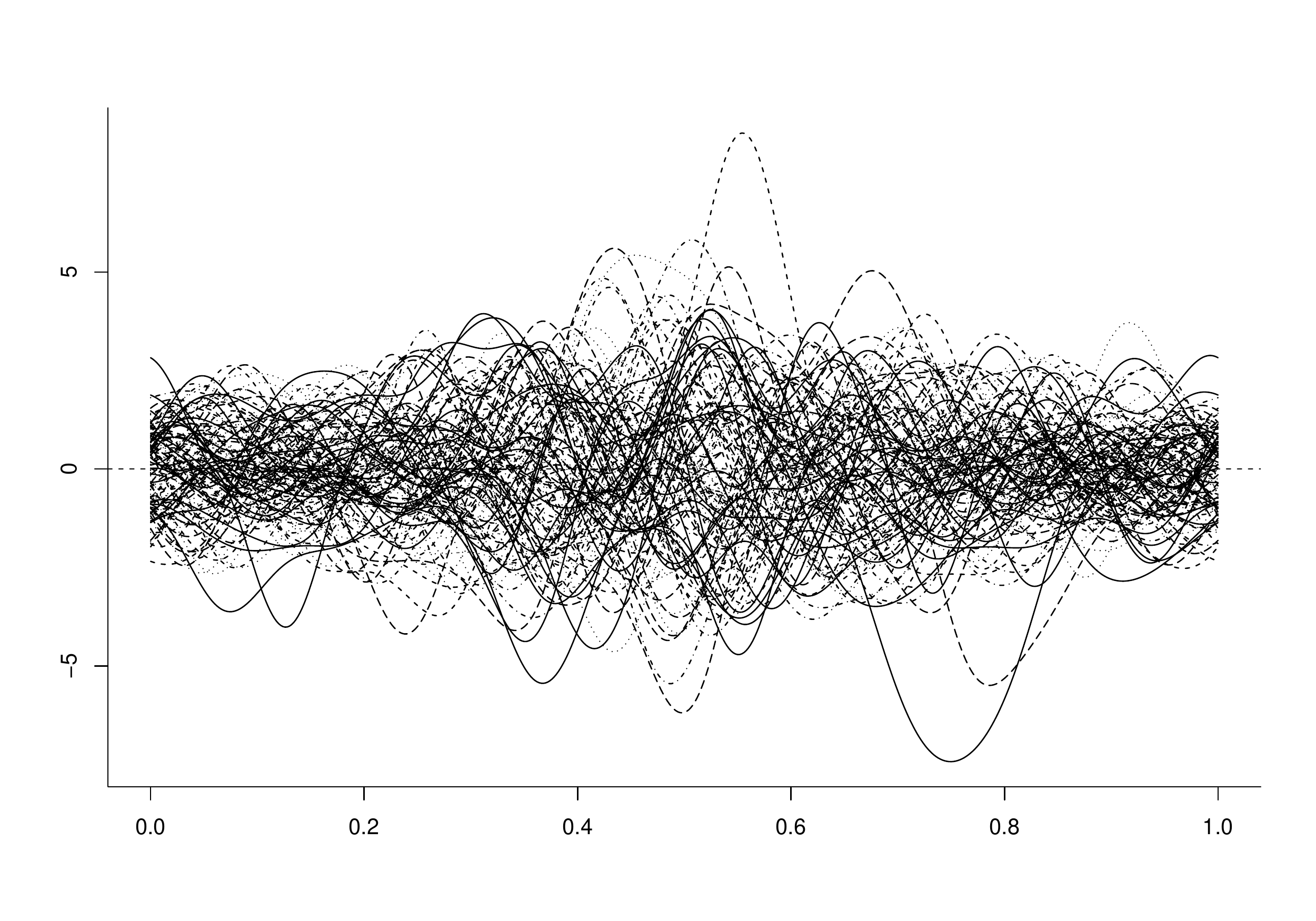} 
\end{center}
\vspace{-0.8cm}
\caption{Upper panel: Time series plot of annual temperature profiles at Gayndah Post Office (left) and scree plot of eigenvalues from the sample covariance operator of the Gayndah Post Office temperature profiles (right). Lower panel: Mean functions of subsegments after binary segmentation using break estimates of mean curves (left) and corresponding demeaned annual temperature curves (right). 
}
\label{Dataplots}
\end{figure}

After detrending, the joint structural break test $J_n(\delta)$ and the trace test $M_n$ were applied. The dimension for jointly testing multiple eigenvalues was chosen based on the total variation explained (TVE) criterion in \eqref{TVE}, setting $v=.85$ so that at least 85\% of the total functional variation was taken into account. As implied by the TVE plot in Figure \ref{Dataplots}, the temperature curves exhibit a slow decay of eigenvalues and $d=10$ was selected. The $p$-value of the joint eigenvalue test was 0.02 
with a break date estimate 1950. 
The test for a structural break in trace led to the same conclusion, the procedure identifying 1950 as the break date estimate. 

\begin{table}[h]
\vspace{.5cm}
\centering
\begin{tabular}{rrrrrrrrr}
\hline
 $j$ & $\hat{\lambda}_j^{(b)}$ & $\hat{\lambda}_j^{(a)}$ & PVE$_j^{(b)}$ & PVE$_j^{(a)}$& TVE$_j^{(b)}$& TVE$_j^{(a)}$ & tr$_j^{(b)}$ & tr$_j^{(a)}$ \\
  \hline
1 & 0.631 & 0.416 & 0.231 & 0.223 & 0.231 & 0.223 & 0.631 & 0.416 \\
  2 & 0.464 & 0.297 & 0.170 & 0.159 & 0.401 & 0.382 & 1.095 & 0.713 \\
  3 & 0.328 & 0.220 & 0.120 & 0.118 & 0.521 & 0.500 & 1.423 & 0.933 \\
  4 & 0.253 & 0.152 & 0.093 & 0.081 & 0.614 & 0.581 & 1.676 & 1.085 \\
  5 & 0.210 & 0.131 & 0.077 & 0.070 & 0.691 & 0.651 & 1.886 & 1.216 \\
  6 & 0.149 & 0.112 & 0.055 & 0.060 & 0.745 & 0.711 & 2.035 & 1.328 \\
  7 & 0.142 & 0.090 & 0.052 & 0.048 & 0.797 & 0.760 & 2.177 & 1.419 \\
  8 & 0.104 & 0.082 & 0.038 & 0.044 & 0.835 & 0.803 & 2.280 & 1.500 \\
  9 & 0.099 & 0.079 & 0.036 & 0.042 & 0.872 & 0.846 & 2.379 & 1.579 \\
  10 & 0.084 & 0.064 & 0.031 & 0.034 & 0.902 & 0.880 & 2.463 & 1.643 \\
   \hline
\end{tabular}
\caption{Table of pre- and post-break eigenvalues $\hat{\lambda}_j^{(b)}$ and $\hat{\lambda}_j^{(a)}$, proportion of variation explained $\mathrm{PVE}_j=\hat{\lambda}_j/\mathrm{tr}_D$, 
total variation explained $\mathrm{TVE}_j=\mathrm{tr}_j/\mathrm{tr}_D$ 
and tr$_j = \sum_{j^\prime=1}^j\hat{\lambda}_{j^\prime}$.
}
\label{B&A table}
\end{table}

It is evident from Table \ref{B&A table} that estimates of all eigenvalues decreased, often dramatically, after the estimated break location in 1950. This decrease also led to a significant structural break in the trace of the covariance operator. While the annual temperature curves had total variation of about $2.46$ degrees Celsius before 1950, this variation subsequently shrank to $1.643$ degrees Celsius. Taking mean function and covariance operator analyses together, it is seen that increasing annual minimum temperature profiles are accompanied by shrinking variation for this data set. To elucidate further,
%
consider the variation explained by the first $d$ eigendirections around the average minimum temperature curves before and after 1950. Total variation around the mean minimum temperature curve before 1950 can be represented as $\hat{\mu}_1\pm\sum_{j=1}^d\sqrt{\hat{\lambda}_j^{(b)}}\hat{\varphi}_j$, the superscript $(b)$ signifying ``before''. Similarly total variation around the most recent average minimum temperature curve can be calculated as $\hat{\mu}_3\pm\sum_{j=1}^d\sqrt{\hat{\lambda}_j^{(a)}}\hat{\varphi}_j$, the superscript $(a)$ signifying ``after''. Here, $\hat{\lambda}_j^{(b)}$ and $\hat{\lambda}_j^{(a)}$ denote the $j$th eigenvalue of the covariance operator of the temperature curves before and after 1950, respectively. It is seen in Figure \ref{TempVar} that the annual minimum temperatures are rising while annual temperature variation is declining. This phenomenon is most pronounced in the months comprising the Australian winter season. As further visual evidence, Figure \ref{TempVar} displays the estimated pre- and post-break covariance kernels. Most of the differences can be seen to be along the diagnoal and during the middle of the year.

\begin{figure}[h]
\vspace{-.4cm}
\begin{center}
\includegraphics[width=.6\textwidth]{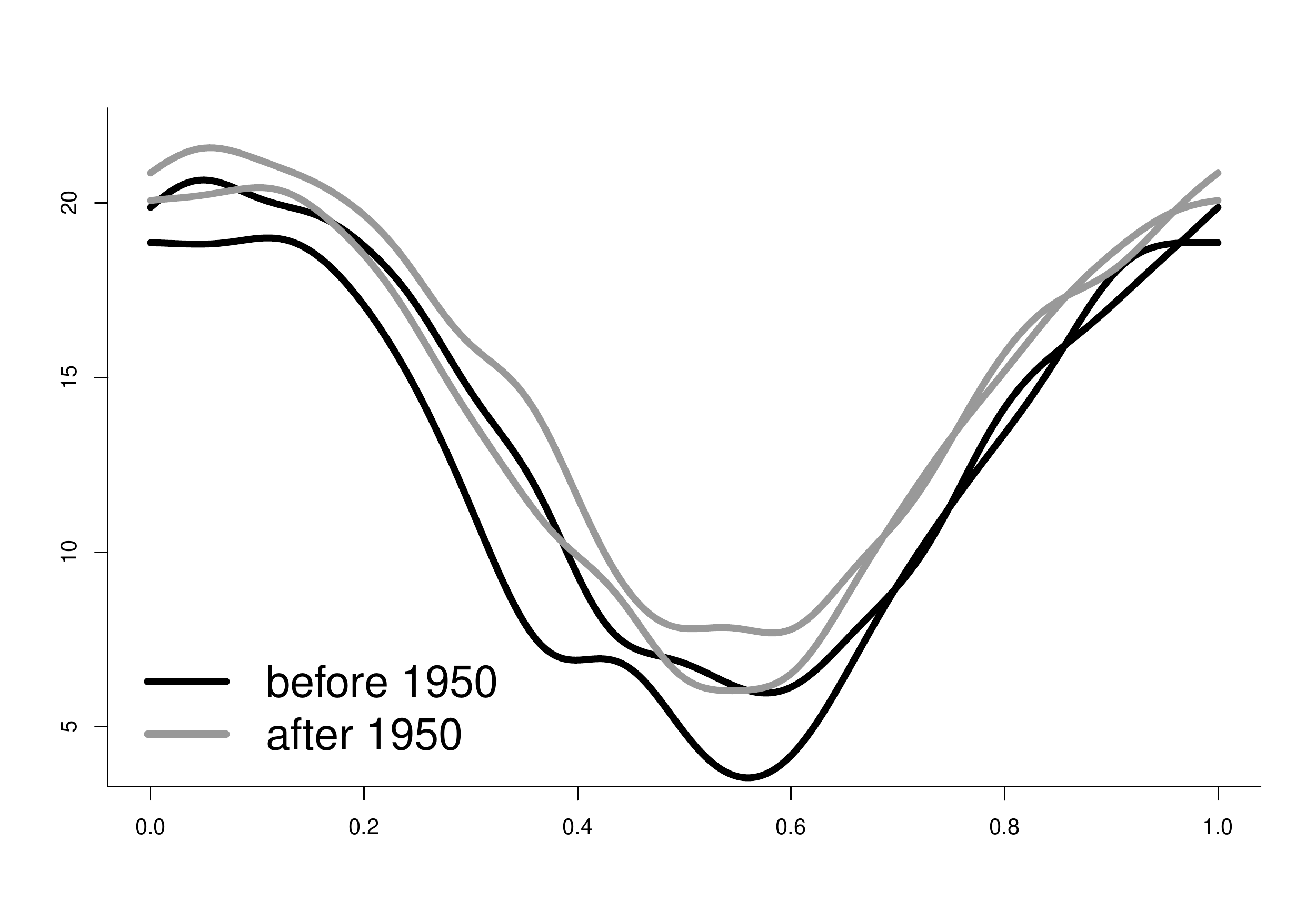} \\
\includegraphics[width=.45\textwidth]{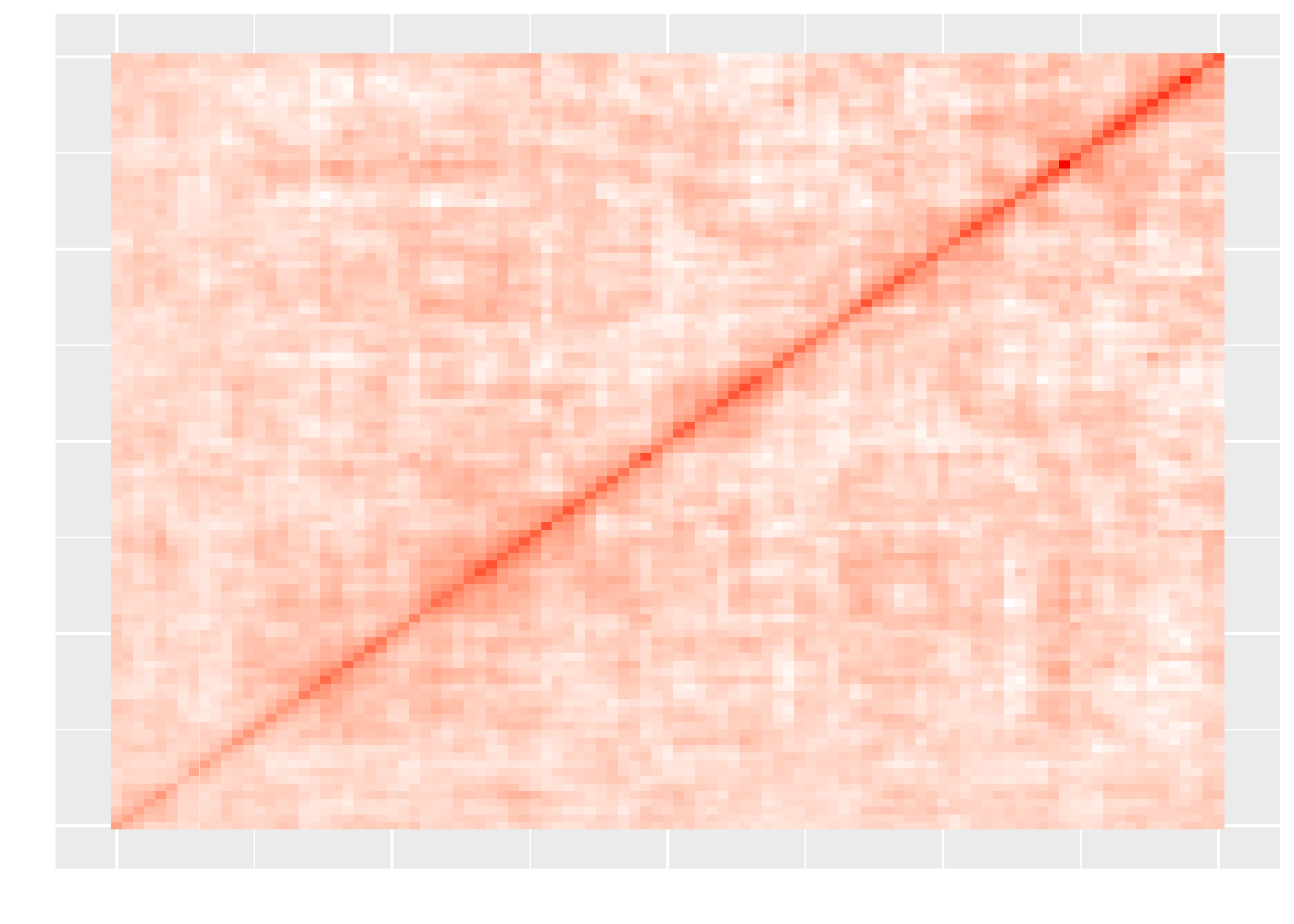}
\includegraphics[width=.45\textwidth]{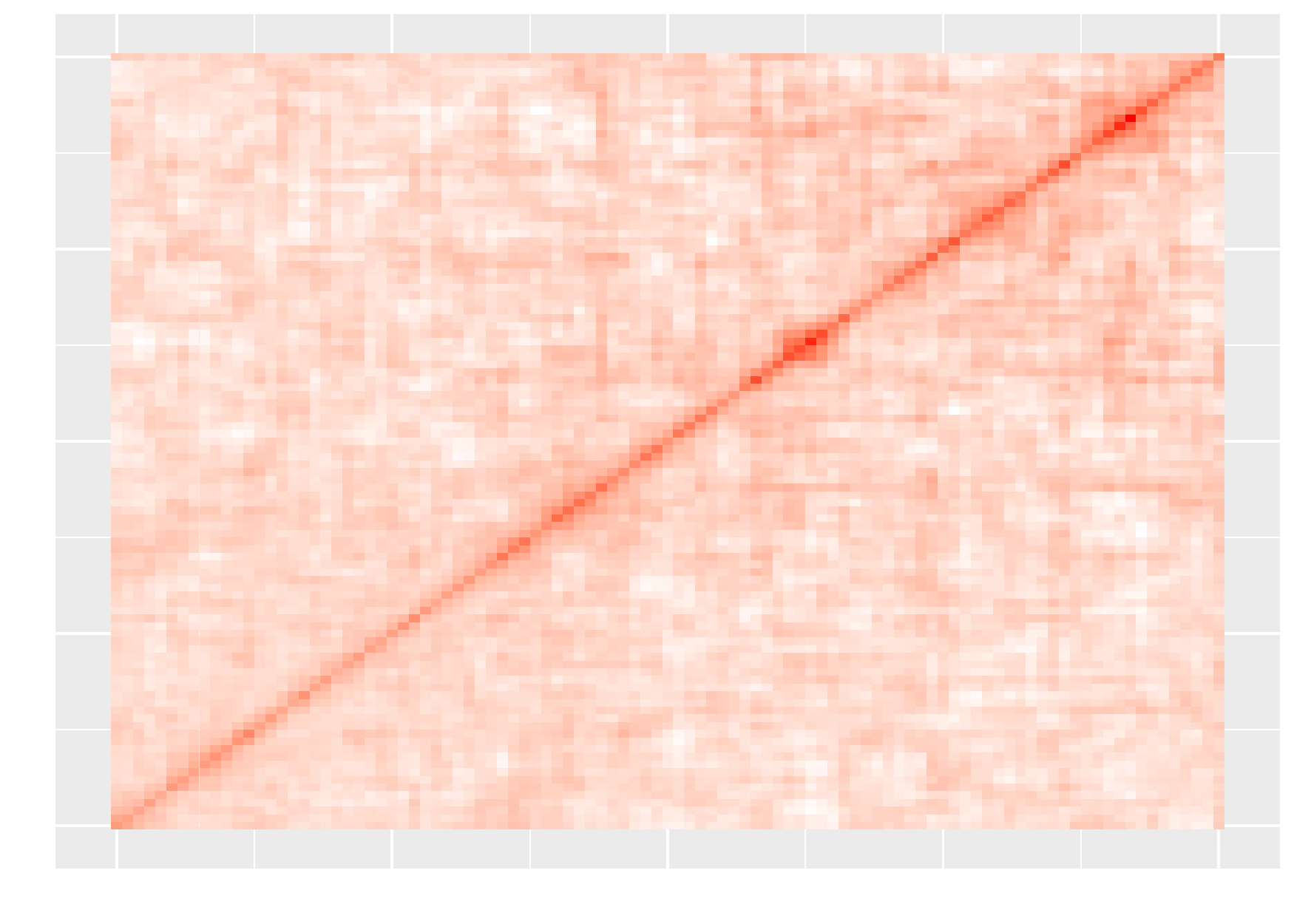}
\end{center}
\vspace{-0.8cm}
\caption{Upper panel: Total variation around the mean annual temperature curve for Gayndah Post Office before (in black) and after (in gray) the estimated break data 1950. Lower panel: Heat map of pre-break (left) and post-break (right) sample covariance kernel restricted to the middle third of the year roughly corresponding to the Australian winter season.}
\label{TempVar}
\end{figure}

\begin{table}[H]
\vspace{.5cm}
\centering
\begin{tabular}{ccc}
  \hline
 Eigenvalue & $p$-value & break date estimate (year) \\
  \hline
  \phantom{1}1 & 0.064 &  1908 \\
  \phantom{1}2 & 0.095 &  1944 \\
  \phantom{1}3 & 0.043 &  1953 \\
  \phantom{1}4 & 0.248 &  1919 \\
  \phantom{1}5 & 0.002 &  1943 \\
  \phantom{1}6 & 0.461 &  1929 \\
  \phantom{1}7 & 0.607 &  1950 \\
  \phantom{1}8 & 0.323 &  1972 \\
  \phantom{1}9 & 0.613 &  1969 \\
  10 & 0.879 &  1906 \\
   \hline
\end{tabular}
\label{DataInd}
\caption{Results for testing changes in a single eigenvalues for the annual temperature curves for Gayndah Post Office.}
\end{table}

A natural follow-up question is if there are any dominant modes of variation driving the observed diminishing variation. To check this, individual tests $I_{j,n}(\delta)$ were applied for $j=1,\ldots,10$. The results are presented in Table \ref{DataInd}. Adjusting nominal levels based on multiple testing, there is some evidence for individual breaks but none, with the possible exception of $j=5$, exerted a dominant influence, indicating that differences across all directions compound to yield the strong rejection of the null observed for the trace test.

The remainder of this section focuses on a short discussion on whether the breaks related to the spectrum of the covariance operator were accompanied by simultaneous breaks in the corresponding  eigenfunctions. Dating and detecting structural breaks in the eigenfunctions, either jointly or marginally, is a rather complicated problem deserving of its own manuscript. Here, the problem will only be briefly approached from the point of view of testing the equality of covariance operators in functional samples, as in Fremdt et al.\ (2012). These authors introduced a two-sample test which obeys a chi-squared asymptotic distribution with known degrees of freedom. To make use of these results in the present analysis, the (joint) effect of breaks in the eigenvalues were taken into account by standardizing the functional sample $X_1,\dots,X_n$ through the transformation
$$
Y_i = 
 \sum_{j=1}^d \frac{1}{\sqrt{\hat{\lambda}_j^{(\ell_i)}}}\langle X_i, \hat{\varphi_j}\rangle\hat{\varphi}_{j},
$$
where $\ell_i=b$ for $i=1,\ldots,k^*$ and $\ell_i=a$ for $i=k^*+1,\ldots,n$, and $\hat{\varphi}_1,\ldots,\hat\varphi_D$ the sample eigenfunctions. The transformed data was then split up into two subsamples using the estimated break data $\hat k^*$ (1950). Since eigenvalue breaks have been removed from $Y_1,\ldots,Y_n$, the two subsamples should have equal covariance structure if there was no break in the eigenfunctions. The test, indeed, yielded a $p$-value of $0.83$ indicating covariance homogeneity. There was thus strong evidence that only the eigenvalues and total variation of the annual minimum temperature curves at Gayndah Post Office were subject to structural breaks but that these breaks did not extend to the eigenfunctions. This indicates stability of seasonal patterns outside those affecting their magnitude. For this particular data set much of the structural break was captured by an increase in minimum temperatures during the Australian winter. It should finally be mentioned that the test of Fremdt et al.\ (2012) was designed for independent Gaussian functions. The authors discussed that in the case of violated normality and independence assumptions, their test was rather conservative in the sense that the likelihood of falsely not rejecting the null hypothesis was narrow. The large $p$-value obtained here adds further support to the conclusion of homogenous eigenfunctions.


\section{Conclusion}\label{sec:conc}
Several methods were proposed for detecting and localizing structural breaks in the covariance operator of a functional time series based on measuring the fluctuations of partial sample estimates of its eigenvalues and trace. Collectively the proposed tests provide a differential procedure for determining how variability in functional time series changes, whether it be in specific eigenvalues, several eigenvalues, or in the trace of the operator. A simulation study showed that these methods perform well even with fairly small samples. In an application to functional data derived from daily minimum temperatures taken in Australia, strong evidence was found that, after taking into account changes in the level, the variability of these curves significantly decreases, and moreover that this change appears to be across all eigenvalues. The change in variability also does not seem to affect the principal components/eigenfunctions, but a rigorous test for changes in the eigenfunctions is left as a possible direction for future research.

\appendix

\section{Proof of Theorems \ref{thm-joint} and \ref{thm-trace}}\label{proofs}

The proof of Theorem \ref{thm-joint} will be developed as a sequence of four lemmas. Throughout $k_i$, $i\ge1$, is used to denote unimportant absolute numeric constants. Under model \eqref{model-1} and Assumption \ref{edep} it may be assumed without loss of generality that $\mu=0$. Define
$$
\tilde{C}_x(t,t^\prime) = \frac{1}{n} \sum_{i=1}^{\lfloor nx \rfloor } X_i(t) X_i(t^\prime),
\qquad x; t,t^\prime\in[0,1].
$$

\begin{lemma}\label{mean-rem}
Under the conditions of Theorem \ref{thm-joint},
$$
\sup_{\delta \le x \le 1} \| \hat{C}_x - \tilde{C}_x \| = O_P\left(\frac{1}{n}\right).
$$
\end{lemma}

\begin{proof}
The proof follows from standard arguments, some of which appear in subsequent lemmas, and so details are omitted.
\end{proof}

For $x\in[0,1]$, let $\tilde{\lambda}_j(x)$ and $\tilde{\varphi}_{j,x}$ denote the ordered eigenvalues and orthonormal eigenfunctions of the integral operator with kernel $\tilde{C}_x$. Since the eigenfunctions $\varphi_j,\hat{\varphi}_{j,x},$ and $\tilde{\varphi}_{j,x}$ are unique only up to a sign, assume without loss of generality that $\langle \varphi_j,\hat{\varphi}_{j,x} \rangle \ge 0$ and $\langle \varphi_j,\tilde{\varphi}_{j,x} \rangle \ge 0$.

\begin{lemma}\label{lam-mr}
Under the conditions of Theorem \ref{thm-joint},
$$
\sup_{\delta \le x \le 1} |\hat{\lambda}_i(x) - \tilde{\lambda}_i(x) | = O_P\left(\frac{1}{n}\right),
$$
for $j=1,\ldots,d$ with $d$ defined in Assumption \ref{d-lam}.
\end{lemma}

\begin{proof}
Lemma 2.2 of Horv\'ath and Kokoszka (2012) 
gives
\begin{align}\label{gohberg}
|\hat{\lambda}_j(x) - \tilde{\lambda}_j(x) | \le  \| \hat{C}_x - \tilde{C}_x \|,
\end{align}
so that the result follows from Lemma \ref{mean-rem}.
\end{proof}

\begin{lemma}\label{tc-rn}
Under the conditions of Theorem \ref{thm-joint},
$$
\sup_{\delta \le x \le 1} \bigg\|  \tilde{C}_x- \frac{\lfloor nx \rfloor }{n} C \bigg\| = O_P\left(\frac{1}{\sqrt{n}}\right).
$$
\end{lemma}

\begin{proof}
By  definition of $\tilde{C}_x$,
\begin{align}\label{l2-1}
\tilde{C}_x(t,t^\prime) - \frac{\lfloor nx \rfloor}{n}C(t,t^\prime)
= \frac{1}{n} \sum_{i=1}^{\lfloor nx \rfloor} \rho_i(t,t^\prime)
\qquad x;t,t^\prime\in[0,1],
\end{align}
where $\rho_i(t,t^\prime) =X_i(t)X_i(t^\prime) - \mathbb{E}[X_0(t)X_0(t^\prime)]$. The Cauchy--Schwarz inequality and Assumption \ref{edep} yield $E[\|\rho_i\|^2] < \infty$, and hence all $\rho_i$ are a.s.\ elements of $L^2([0,1]^2)$. Define
$$
\rho^{(m)}_i(t,t^\prime) = X_{i,m}(t)X_{i,m}(t^\prime) - \mathbb{E}[X_0(t)X_0(t^\prime)],
\qquad t,t^\prime\in[0,1].
$$
Then, ($\rho_i^{(m)})$ is a centered, $m$-dependent sequence. If $q=p/2$ with $p$ given in Assumption \ref{edep}, then triangle and Minkowski's inequalities imply
\begin{align}\label{l2-2}
\big(\mathbb{E}[\|\rho_i - \rho_i^{(m)} \|^q]\big)^{1/q}
&\le \big( \mathbb{E}[\{ \|X_i\otimes(X_i-X_{i,m})\| + \|(X_i-X_{i,m})\otimes X_{i,m}\| \}^q]\big)^{1/q} \\
&\le  \big( \mathbb{E}[ \|X_i\otimes(X_i-X_{i,m})\|^q]\big)^{1/q}
+ \big(\mathbb{E}[\|(X_i-X_{i,m})\otimes X_{i,m}\|^q]\big)^{1/q}. \notag
\end{align}
The definition of the norm in $L^2([0,1])$ gives $\|X_i\otimes(X_i-X_{i,m})\|= \|X_i\|\|X_i-X_{i,m}\|$. Hence, Cauchy--Schwarz inequality and stationarity yield that the right-hand side of \eqref{l2-2} is upper-bounded by
\begin{align*}
(\mathbb{E}[\|X_0\|^{2q}])^{1/2q}(\mathbb{E}[\|X_0-X_{0,m}\|^{2q}])^{1/2q}.
\end{align*}
Taken together it follows that
\begin{align}\label{rho-s}
\sum_{m=1}^\infty (\mathbb{E}[\|\rho_i - \rho_i^{(m)}\|^q])^{1/q}
\le (\mathbb{E}[\|X_0\|^{p}])^{1/p} \sum_{m=1}^\infty (E[\|X_0-X_{0,m}\|^{p}])^{1/p} < \infty
\end{align}
according to Assumption \ref{edep}. Examining the proofs of Proposition 4 and Corrollary 1 in Berkes et al. (2011), it is seen that their results hold for arbitrary separable Hilbert space-valued random variables $\rho_i$ defined as Bernoulli shifts satisfying \eqref{rho-s}. Thus,
$$
E\bigg[\bigg(\sup_{0 \le x \le 1} \bigg\| \sum_{i=1}^{\lfloor nx \rfloor} \rho_i \bigg\| \bigg)^q\bigg]
\le k_0 n^{q/2}.
$$
Combining this result with \eqref{l2-1} and Chebyshev's inequality implies the assertion of the lemma.
\end{proof}

\begin{lemma}\label{main-approx}
Under the conditions of Theorem \ref{thm-joint}, for all $x\in [0,1]$,
$$
\tilde{\lambda}_j(x) - \frac{\lfloor nx \rfloor}{n} \lambda_j =
\bigg\|\bigg(\tilde C_x-\frac{\lfloor nx\rfloor}n C\bigg)\varphi_j\otimes\varphi_j\bigg\|^2
+ R_{j,n}(x),
$$
where
$$
\sup_{\delta \le x \le 1 } |R_{j,n}(x)| = O_P\left(\frac{1}{n}\right).
$$
\end{lemma}

\begin{proof}
A direct calculation using the definitions of $\tilde{\lambda}_j(x)$, $\tilde{\varphi}_{j,x}$, $\lambda_j$, and $\varphi_j$ shows that, for $t\in[0,1]$,
\begin{align}\label{l3-0}
\int \bigg\{\frac{\lfloor nx \rfloor}{n}C(t,t^\prime) & + \bigg(\tilde{C}_x(t,t^\prime) - \frac{\lfloor nx \rfloor}{n}C(t,t^\prime) \bigg) \bigg\}\{\varphi_j(t^\prime)+ (\tilde{\varphi}_{j,x}(t^\prime)-\varphi_j(t^\prime))\}dt^\prime \\
&= \bigg\{\frac{\lfloor nx \rfloor}{n}\lambda_j  + \bigg(\tilde{\lambda}_j(x)-\frac{\lfloor nx \rfloor}{n}\lambda_j\bigg)\bigg\}\{\varphi_j(t)+ (\tilde{\varphi}_{j,x}(t)-\varphi_j(t))\}. \notag
\end{align}
Therefore,
\begin{align}\label{l3-1}
\bigg(\tilde{\lambda}_j(x)-\frac{\lfloor nx \rfloor}{n}\lambda_j  \bigg)\varphi_j(t)
&= \int \bigg(\tilde{C}_x(t,t^\prime) - \frac{\lfloor nx \rfloor}{n}C(t,t^\prime) \bigg) \varphi_j(t^\prime)dt^\prime \\
&+ \int \frac{\lfloor nx \rfloor}{n} C(t,t^\prime) (\tilde{\varphi}_{j,x}(t^\prime)-\varphi_j(t^\prime))dt^\prime - \frac{\lfloor nx \rfloor}{n}\lambda_j (\tilde{\varphi}_{j,x}(t)-\varphi_j(t)) + G_{j,n}(t,x),\notag
\end{align}
where
\begin{align*}
G_{j,n}(t,x) = \int \bigg(\tilde{C}_x(t,t^\prime) - &\frac{\lfloor nx \rfloor}{n}C(t,t^\prime) \bigg) (\tilde{\varphi}_{j,x}(t^\prime)-\varphi_j(t^\prime)) dt^\prime \\
&- \bigg(\tilde{\lambda}_j(x)-\frac{\lfloor nx \rfloor}{n}\lambda_j\bigg)(\tilde{\varphi}_{i,x}(t)-\varphi_i(t)). \notag
\end{align*}
Let
$R_{j,n}(x) = \langle \varphi_j(\cdot),G_{j,n}(\cdot,x)\rangle$ 
It follows from the triangle and Cauchy--Schwarz inequalities that
\begin{align}\label{l3-3}
|R_{j,n}(x)| \le \bigg\| \tilde{C}_x - \frac{\lfloor nx \rfloor}{n}C \bigg\| \|\tilde{\varphi}_{j,x}-\varphi_j\| + |\tilde{\lambda}_j(x)-\frac{\lfloor nx \rfloor}{n}\lambda_j|\|\tilde{\varphi}_{j,x}-\varphi_j\|.
\end{align}
By Assumption \ref{d-lam}, there exists a constant $k_1>0$ such that $ \min_{1\le j \le d} x(\lambda_j - \lambda_{j+1}) \ge k_1$ for all $x\in[\delta,1]$. It follows from Lemmas 2.2 and 2.3 in Horv\'ath and Kokoszka (2012) in combination with the choice of the sign of $\tilde{\varphi}_{j,x}$ that
\begin{align}\label{l3-3.5}
\|\tilde{\varphi}_{j,x}-\varphi_j\| \le \frac{1}{k_1} \bigg\| \tilde{C}_x - \frac{\lfloor nx \rfloor}{n}C \bigg\|
\qquad\mbox{and}\qquad
\bigg|\tilde{\lambda}_j(x)-\frac{\lfloor nx \rfloor}{n}\lambda_j\bigg| \le \bigg\| \tilde{C}_x - \frac{\lfloor nx \rfloor}{n}C \bigg\|.
\end{align}
Taken together with \eqref{l3-3} and Lemma \ref{tc-rn} this gives
\begin{align}\label{l3-4}
\sup_{\delta \le x \le 1} |R_{j,n}(x)| \le k_2 \bigg\| \tilde{C}_x - \frac{\lfloor nx \rfloor}{n}C \bigg\|^2
= O_P\left(\frac{1}{n}\right).
\end{align}
Returning to equation \eqref{l3-1}, taking the inner product of the left- and right-hand sides with $\varphi_j(t)$ implies
\begin{align}\label{l3-5}
\tilde{\lambda}_j(x)-\frac{\lfloor nx \rfloor}{n}\lambda_j
=& \intt \bigg(\tilde{C}_x(t,t^\prime) - \frac{\lfloor nx \rfloor}{n}C(t,t^\prime) \bigg) \varphi_j(t^\prime)\varphi_j(t)dtdt^\prime \\
&+ \intt \frac{\lfloor nx \rfloor}{n} C(t,t^\prime)\varphi_j(t) (\tilde{\varphi}_{j,x}(t^\prime)-\varphi_j(t^\prime))dtdt^\prime \notag \\
&- \int \frac{\lfloor nx \rfloor}{n}\lambda_j (\tilde{\varphi}_{j,x}(t)-\varphi_j(t))\varphi_j(t)dt + R_{j,n}(x).\notag
\end{align}
Noticing that $\int C(t,t^\prime) \varphi_j(t)dt = \lambda_j \varphi_j(s)$, the second and third terms on the right-hand side of \eqref{l3-5} negate each other, and the lemma follows from \eqref{l3-4}.
\end{proof}

\begin{lemma}\label{theta-sum}
If $\theta_{i,j}^{(m)}=\langle\rho_i^{(m)},\varphi_j\otimes\varphi_j\rangle$, 
then for $q=p/2$,
$$
\sum_{m=1}^\infty \big(\mathbb{E}\big[|\theta_{0,j} - \theta_{0,j}^{(m)}|^q\big]\big)^{1/q} < \infty.
$$
\end{lemma}

\begin{proof}
By the Cauchy--Schwarz inequality,
\begin{align*}
|\theta_{0,j} - \theta_{0,j}^{(m)}|
&= \big|\langle\rho_i-\rho_i^{(m)},\varphi_j\otimes\varphi_j\rangle\big| 
\le \|\rho_0 - \rho_0^{(m)} \|.
\end{align*}
This implies with \eqref{rho-s} that
$$
\sum_{m=1}^\infty \big(\mathbb{E}\big[|\theta_{0,j} - \theta_{0,j}^{(m)}|^q\big]\big)^{1/q}
\le \sum_{m=1}^\infty \big(\mathbb{E}\big[\|\rho_0 - \rho_0^{(m)} \|^q\big]\big)^{1/q} < \infty.
$$

\end{proof}

\medskip
%
%

The next aim is to establish Theorem \eqref{thm-joint}. Let $\|\cdot\|_d$ denote standard Euclidean norm in $\mathbb{R}^d$, that is, for $y \in \mathbb{R}^d$,
$
\|y\|_d = ( \sum_{j=1}^{d} y_j^2 )^{1/2}.
$
Let also
$
\tilde{{\bf \Lambda}}_d(x) = ( \tilde{\lambda}_1(x),\ldots,\tilde{\lambda}_d(x))^\top \in D^d[\delta,1].
$

\begin{lemma}
\label{lem-1-thm-2}
Under the conditions of Theorem \ref{thm-joint},
$$
\sup_{\delta \le x \le 1} \| \hat{{\bf \Lambda}}_d(x) -\tilde{{\bf \Lambda}}_d(x) \|_d = O_P\left( \frac{1}{n}\right).
$$
\end{lemma}
\begin{proof}
Using \eqref{gohberg}, it follows that
\begin{align*}
\sup_{\delta \le x \le 1} \| \hat{{\bf \Lambda}}_d(x) -\tilde{{\bf \Lambda}}_d(x) \|_d
= \sup_{\delta \le x \le 1} \bigg( \sum_{j=1}^{d} |\hat{\lambda}_j - \tilde{\lambda}_j|^2 \bigg)^{1/2}
\le \sup_{\delta \le x \le 1} \|\hat{C}_x - \tilde{C}_x\|,
\end{align*}
and so the result follows from Lemma \ref{mean-rem}.
\end{proof}

Let
$
{\bf R}_{d,n}(x) = (R_{1,n}(x),\ldots,R_{d,n}(x))^\top \in \mathbb{R}^d,
$
where $R_{j,n}(x)$ is defined in Lemma \ref{main-approx}.

\begin{lemma}\label{lem-2-thm-2}
Under the conditions of Theorem \eqref{thm-joint},
$$
\tilde{{\bf \Lambda}}_d(x) - \frac{\lfloor n x \rfloor}{n} {\bf \Lambda}_d = \frac{1}{n} \sum_{i=1}^{\lfloor nx \rfloor } {\bf \Theta}_i+ {\bf R}_{d,n}(x),
$$
where ${\bf \Theta}_i$ are defined after Theorem \ref{thm-joint} and
\begin{align}\label{r-vec-approx}
  \sup_{\delta \le x \le 1} \|{\bf R}_{d,n}(x)\|_d = O_P\left(\frac{1}{n}\right).
\end{align}
\end{lemma}

\begin{proof}
Considering \eqref{l3-5} for $j=1,\ldots,d$ implies
$$
\tilde{{\bf \Lambda}}_d(x) - \frac{\lfloor n x \rfloor}{n} {\bf \Lambda}_d = \frac{1}{n} \sum_{i=1}^{\lfloor nx \rfloor } {\bf \Theta}_i + {\bf R}_{d,n}(x).
$$
Furthermore, it follows from \eqref{l3-4} that
$$
\max_{1\le j \le d} \sup_{\delta \le x \le 1} |R_{j,n}(x)| = O_P\left( \frac{1}{n} \right),
$$
from which the result follows.
\end{proof}

%
%
%

\begin{proof}[Proof of Theorem \ref{thm-joint}] Let
$
{\bf \Theta}_i^{(m)} = (\theta_{i,1}^{(m)},\ldots,\theta_{i,d}^{(m)})^\top,
$
where $\theta_{i,j}^{(m)}$ is defined in Lemma \ref{theta-sum}. Lyapounov's inequality yields
\begin{align*}
\big(\mathbb{E}\big[\|{\bf \Theta}_0 - {\bf \Theta}_0^{(m)}\|_d^2\big]\big)^{1/2}
&= \bigg( \sum_{j=1}^{d} \mathbb{E}\big[|\theta_{0,j} - \theta_{0,j}^{(m)}|^2\big] \bigg)^{1/2} \\
& \le d \max_{1 \le j \le d} \big(\mathbb{E}\big[|\theta_{0,j} - \theta_{0,j}^{(m)}|^2\big] \big)^{1/2} \\
& \le d \max_{1 \le j \le d} \big(\mathbb{E}\big[|\theta_{0,j} - \theta_{0,j}^{(m)}|^q\big] \big)^{1/q}.
\end{align*}
Therefore, by Lemma \ref{theta-sum},
\begin{align}\label{theta-approx-vec}
\sum_{m=1}^{\infty} \big(\mathbb{E}\big[\|{\bf \Theta}_0 - {\bf \Theta}_0^{(m)}\|_d^2\big] \big)^{1/2} < \infty.
\end{align}
Now the assertion follows from Lemmas \ref{lem-1-thm-2} and \ref{lem-2-thm-2}, and Theorem A.1 of Aue et al. (2009), noting that the sum defining the matrix
$$
\Sigma_d(j,j^\prime)= \sum_{i=-\infty}^\infty \mathrm{Cov}(\theta_{0,j},\theta_{i,j}),
\qquad j,j^\prime=1,\ldots, d,
$$
converges pointwise absolutely.
\end{proof}

The proof of Theorem \ref{thm-trace} is similar to the above results. A sketch of the basic idea is given below.

\begin{proof}[Sketch of the proof of Theorem \ref{thm-trace}] Let $
\tilde{T}_n(x) 
= \frac{1}{n} \sum_{i=1}^{\lfloor nx \rfloor } \tilde{\xi}_i$, where $\xi_i=\|X_i-\mu\|^2$.
It follows then along the lines of the proof of Lemma \ref{lam-mr} that
$$
\sup_{0 \le x \le 1} |\hat{T}_n(x) - \tilde{T}_n(x) | = o_P(1).
$$
With $\tilde{\xi}_i^{(m)}=\|X_i^{(m)}-\mu\|^2$, 
it follows as in Lemma \ref{theta-sum} that
$$
\sum_{m=1}^{\infty} \big(\mathbb{E}\big[(\xi_0 - \xi_0^{(m)})^2\big]\big)^{1/2} < \infty.
$$
Now the theorem follows from Theorem 3 of Wu (2005).
\end{proof}

\section{Estimation of $\Sigma_d$ and $\sigma_T^2$, and proof of Theorem \ref{thm-cons}} \label{var-est}


Let $\hat{\theta}_{i,j}$ be defined in \eqref{theta-def-est}, and let
$
\hat{{\bf \Theta}}_i = (\hat{\theta}_{i,1},\ldots,\hat{\theta}_{j,d})^\top.
$
The estimator for the long-run covariance matrix $\Sigma_d$ is given by
$$
\hat{\Sigma}_d = \sum_{\ell= -\infty}^{\infty} w\left(\frac{\ell}{h}\right) \hat{{\bf \Gamma}}_{\ell,\theta},
$$
where $h$ is a bandwidth parameter satisfying $h=h(n)$, and $1/h(n) + h(n)/n^{1/2} \to 0 $ as $n\to \infty$, and
\begin{align*}
 \hat{{\bf \Gamma}}_{\ell,\theta}=
\frac{1}{n}\sum_{i\in\mathcal{I}_\ell}\big[\hat{{\bf \Theta}}_i-\bar{{\bf \Theta}}\big]\big[\hat{{\bf \Theta}}_{i+\ell}-\bar{{\bf \Theta}}\big]^\top,
\end{align*}
with $\mathcal{I}_\ell=\{1,\ldots,n-\ell\}$ if $\ell\geq 0$ and $\mathcal{I}_\ell=\{1-\ell,\ldots,n\}$ if $\ell<0$, as well as
$   \bar{{\bf \Theta}}= \frac{1}{n} \sum_{i=1}^{n}\hat{{\bf \Theta}}_i$,
whereas $w$ is a symmetric weight function with bounded support satisfying the standard conditions $w(0)=1$, $w(u)=w(-u)$, $w(u)\le 1$, $w(u)=0$ if $|u|>m$ for some $m>0$, $w$ is continuous, and for some $b>0$
\begin{align}\label{order}
0<  \mathfrak{q} =\lim_{x\to 0} x^{-b}[1-w(x)] < \infty.
\end{align}

\begin{theorem}\label{sig-cons-thm}
Under the conditions of Theorem \ref{thm-joint}, $\hat{\Sigma}_d$ satisfies \eqref{sig-cons}.
\end{theorem}

\begin{proof}
Let
$$
\hat{\Sigma}'_d = \sum_{\ell= -\infty}^{\infty} w\left(\frac{\ell}{h}\right) \hat{{\bf \Gamma}}'_{\ell,\theta},
\qquad\mbox{and}\qquad
\tilde{\Sigma}_d = \sum_{\ell= -\infty}^{\infty} w\left(\frac{\ell}{h}\right) \tilde{{\bf \Gamma}}_{\ell,\theta},
$$
where
$$
\hat{{\bf \Gamma}}'_{\ell,\theta}=
\frac{1}{n}\sum_{i\in\mathcal{I}_\ell}\big[\hat{{\bf \Theta}}'_i-\bar{{\bf \Theta}}'\big]\big[\hat{{\bf \Theta}}'_{i+\ell}-\bar{{\bf \Theta}}'\big]^\top
\qquad\mbox{and}\qquad
\tilde{{\bf \Gamma}}_{\ell,\theta}=
\frac{1}{n}\sum_{i\in\mathcal{I}_\ell}\big[{\bf \Theta}_i-\bar{{\bf \Theta}}^*\big]\big[{\bf \Theta}_{i+\ell}-\bar{{\bf \Theta}}^*\big]^\top,
$$
with $\hat{{\bf \Theta}}'_j=(\hat{\theta}_{1,j}',\ldots,\hat{\theta}_{d,j}')^\top$,
$
\hat{\theta}_{i,j}' =\langle X_i\otimes X_i-\mathbb{E}[X_0\otimes X_0],\hat\varphi_{j,1}\otimes\hat\varphi_{j,1}\rangle
$, 
$
\bar{{\bf \Theta}}' = \frac{1}{n}\sum_{i=1}^{n}\hat{{\bf \Theta}}'_i$ and $\bar{{\bf \Theta}}^* = \frac{1}{n}\sum_{i=1}^{n}{\bf \Theta}_i$.
Elementary arguments show that
\begin{align}\label{sig-cons-0}
|\hat{\Sigma}'_d - \hat{\Sigma}_d |_F = o_P(1),
\end{align}
and so details are omitted. The first aim is then to show that
\begin{align}\label{sig-cons-1}
  |\tilde{\Sigma}_d - \hat{\Sigma}_d |_F = o_P(1).
\end{align}
To this end, note that
\begin{align*}
\hat{{\bf \Gamma}}_{\ell,\theta} - \tilde{{\bf \Gamma}}_{\ell,\theta}
&= \frac{1}{n} \sum_{i \in \mathcal{I}_\ell } \Big\{ \big[\hat{{\bf \Theta}}_i-\bar{{\bf \Theta}}\big]\big[\hat{{\bf \Theta}}_{i+\ell}-\bar{{\bf \Theta}}\big]^\top  - \big[{\bf \Theta}_i-\bar{{\bf \Theta}}^*\big]\big[{\bf \Theta}_{i+\ell}-\bar{{\bf \Theta}}^*\big]^\top \\
   &\qquad\qquad+ \big[{{\bf \Theta}}_i-\bar{{\bf \Theta}}^*\big]\big[\hat{{\bf \Theta}}_{i+\ell}-\bar{{\bf \Theta}}\big]^\top -\big[{{\bf \Theta}}_i-\bar{{\bf \Theta}}^*\big]\big[\hat{{\bf \Theta}}_{i+\ell}-\bar{{\bf \Theta}}\big]^\top  \Big\}  \\
   &= B_1 + B_2,
\end{align*}
where
\begin{align*}
B_1 &= \frac{1}{n} \sum_{i \in \mathcal{I}_\ell } \big[\hat{{\bf \Theta}}_i-\bar{{\bf \Theta}} - {\bf \Theta}_i+\bar{{\bf \Theta}}^*\big]\big[\hat{{\bf \Theta}}_{i+\ell}-\bar{{\bf \Theta}}\big]^\top, \\
B_2 &= \frac{1}{n} \sum_{i \in \mathcal{I}_\ell } \big[{{\bf \Theta}}_i-\bar{{\bf \Theta}}^*\big]\big[\hat{{\bf \Theta}}_{i+\ell}-\bar{{\bf \Theta}} - {\bf \Theta}_{i+\ell}+\bar{{\bf \Theta}}^*\big]^\top.
\end{align*}
One then has that
\begin{align*}
\Big|\big[\hat{{\bf \Theta}}_i-\bar{{\bf \Theta}} - {\bf \Theta}_i+\bar{{\bf \Theta}}^*\big]
\big[\hat{{\bf \Theta}}_{i+\ell}-\bar{{\bf \Theta}}\big]^\top \Big|_F
\le k_3 \big\|\hat{{\bf \Theta}}_j-\bar{{\bf \Theta}} - {\bf \Theta}_j+\bar{{\bf \Theta}}^*\big\|_d
\big\| \hat{{\bf \Theta}}_{j+\ell}-\bar{{\bf \Theta}} \big\|_d.
\end{align*}
By the triangle inequality,  $ \|\hat{{\bf \Theta}}_j-\bar{{\bf \Theta}} - {\bf \Theta}_j+\bar{{\bf \Theta}}^*\|_d \le \|\hat{{\bf \Theta}}_j - {\bf \Theta}_j\|_d +\|\bar{{\bf \Theta}}-\bar{{\bf \Theta}}^*\|_d$. The Cauchy--Schwarz inequality along with \eqref{l3-3.5} imply that
$$
\|\hat{{\bf \Theta}}_i - {\bf \Theta}_i\|_d  \le k_4 \max_{1\le j \le d} \| \varphi_j - \hat{\varphi}_{j,1}\|  \le k_4 \|\hat{C} - C\|.
$$
Similarly,
$
\|\hat{{\bf \Theta}}_j - {\bf \Theta}_j\|_d  \le k_5  \|\hat{C} - C\|.
$
Several applications of the Cauchy--Schwarz inequality yield
$$
\mathbb{E}\big[ \| \hat{{\bf \Theta}}_{i+\ell}-\bar{{\bf \Theta}} \|_d^2\big] \le k_6 E[\|X_0\|^4].
$$
Combining the above with Lemma \ref{tc-rn}, it follows that
\begin{align*}
\mathbb{E}[|B_1|_F]
&\le  \frac{1}{n} \sum_{i \in \mathcal{I}_\ell }
\mathbb{E}\Big| \big[\hat{{\bf \Theta}}_i-\bar{{\bf \Theta}} - {\bf \Theta}_i+\bar{{\bf \Theta}}^*\big]\big[\hat{{\bf \Theta}}_{i+\ell}-\bar{{\bf \Theta}}\big]^\top \Big|  \\
&\le \frac{1}{n} \sum_{i \in \mathcal{I}_\ell } k_7 \mathbb{E}\big[ \|\hat{C} - C \| \|\hat{{\bf \Theta}}_{i+\ell}-\bar{{\bf \Theta}} \|_d\big] \\
&\le \frac{1}{n} \sum_{i \in \mathcal{I}_\ell }  k_7  \big(\mathbb{E}[ \|\hat{C} - C \|^2]\big)^{1/2}  \big(\mathbb{E}[\hat{{\bf \Theta}}_{i+\ell}-\bar{{\bf \Theta}} \|_d^2] \big) \\
&\le k_8 \big(\mathbb{E}[ \|\hat{C} - C\|^2]\big)^{1/2} \\
&= O\bigg(\frac{1}{\sqrt{n}}\bigg).
\end{align*}
Hence, by Markov's inequality, $|B_1|_F= O_P(1/\sqrt{n})$. Similarly, $|B_2|_F= O_P(1/\sqrt{n})$. Now it follows by the assumptions that $w$ is bounded with bounded support and $h/\sqrt{n} \to 0$ as $n\to \infty$ that
$$
|\tilde{\Sigma}_d - \hat{\Sigma}_d |_F
\le \sum_{\ell=-\infty}^{\infty} w\left(\frac{\ell}{h} \right)| \hat{{\bf \Gamma}}_{\ell,\theta} - \tilde{{\bf \Gamma}}_{\ell,\theta} |_F
=O_P\bigg(\frac{h}{\sqrt{n}}\bigg) = o_P(1),
$$
as desired. The same arguments presented in Chapter 11 of Brockwell and Davis (2006) lead to
\begin{align*}
  |\tilde{\Sigma}_d - \Sigma_d |_F = o_P(1),
\end{align*}
from which the result follows in light of \eqref{sig-cons-0}, \eqref{sig-cons-1}, and the triangle inequality for the Frobenius norm.
\end{proof}

In order to estimate $\sigma_T^2$ in \eqref{sig-trace-def}, use the estimator
$$
\hat{\sigma}_T^2 = \sum_{\ell=-\infty}^{\infty} w \left( \frac{\ell}{h} \right) \hat{\gamma}_\ell,
$$
where
$$
\hat{\gamma}_\ell = \frac{1}{n} \sum_{i\in I_\ell} (\hat{\xi}_i - \bar{\xi})(\hat{\xi}_{i+\ell} - \bar{\xi}),
$$
where $I_\ell = \{1,\ldots,n-\ell\}$ if $\ell \ge 1$, $I_\ell = \{1-\ell,\ldots,n\}$ if $\ell \le 0$, and
$\bar{\xi} = \frac{1}{n} \sum_{i=1}^{n}\hat{\xi}_i$.
It follows as in the proof of Theorem \ref{sig-cons-thm} that under the conditions of Theorem \ref{thm-trace}, $|\sigma_T^2 - \hat{\sigma}_T^2|=o_P(1)$.

\begin{proof}[Proof of Theorem \ref{thm-cons}:]
Begin with establishing part (a). Since $\delta \le \tau$,
\begin{align}\label{j-approx}
J_n(\delta) \ge {\bf \kappa}_n(\tau)^\top \hat{\Sigma}_d^{-1} {\bf \kappa}_n(\tau).
\end{align}
Under Assumption \ref{alt-as}, it follows from Theorem \ref{thm-joint} that there exists a $d$-dimensional vector-valued process $({\bf R}_{d,n}'(\tau)\colon\tau\in[0,1])$ such that
\begin{align}\label{con-thm-1}
\hat{{\bf \Lambda}}_d(\tau) = \tau {\bf \Lambda}_d^{(1)} + {\bf R}_{1,d,n}(\tau),
\end{align}
with ${\bf R}_{d,n}'(\tau)$ satisfying
\begin{align}\label{con-thm-2}
\max_{1\le j \le d} \sup_{\delta \le x \le 1} |{\bf R}_{1,d,n}'(\tau,j)|
= O_P\bigg( \frac{1}{\sqrt{n}} \bigg),
\end{align}
where ${\bf R}_{d,n}'(\tau,j)$ denotes the $j$th coordinate of ${\bf R}_{d,n}'(\tau)$. Moreover, it follows under Assumption \ref{alt-as} that
\begin{align}\label{c-alt-approx}
\hat{C}_1(t,s) = \tau C_1(t,s) + (1-\tau)C_2(t,s) + V_n(t,s),
\end{align}
with $\|V_n\|=O_P(1/\sqrt{n})$. Thus, \eqref{gohberg} implies that
$$
\hat{\lambda}_j(1)
= \tau \lambda_j^{(1)} + (1-\tau) \lambda_j^{(2)} + O_P\bigg(\frac{1}{\sqrt{n}}\bigg).
$$
Hence, \eqref{con-thm-1} and \eqref{con-thm-2} show that
$$
\hat{\lambda}_1(\tau) - \frac{\lfloor n \tau \rfloor }{n} \hat{\lambda}_1(1) = \tau(1-\tau) (\lambda_1^{(1)} - \lambda_1^{(2)})+O_P\bigg(\frac{1}{\sqrt{n}}\bigg),
$$
and so
\begin{align}\label{kap-aprox}
{\bf \kappa}_n(\tau) =\sqrt{n} \tau (1-\tau)  ({\bf \Lambda}_d^{(1)}-  {\bf \Lambda}_d^{(2)})  + {\bf R}_{2,d,n}(\tau), \qquad \|{\bf R}_{2,d,n}(\tau)\|_d = O_P(1).
\end{align}
This provides an approximation for the behavior of ${\bf \kappa}_n(\tau)$ under Assumption \ref{alt-as}. Turning to the asymptotic behavior of $\hat{\Sigma}_d^{-1}$ under $H_A$ and Assumption \ref{alt-as}, note that imposing Assumption \ref{d-lam} on the models $g_1$ and $g_2$ ensures that the eigenvalues of the integral operator with kernel $C^*(t,t^\prime) = \tau C_1(t,t^\prime) + (1-\tau)C_2(t,t^\prime)$ have strictly positive spacings, and hence the same perturbation result used to establish \eqref{l3-3.5} implies
\begin{align}\label{var-phi-alt}
\| \hat{\varphi}_{j,1} - \varphi_j^* \| \le k_9 \| \hat{C}_1 - C^*\| = O_P\bigg(\frac{1}{\sqrt{n}}\bigg).
\end{align}

Let ${\bf \Theta}^{(k)}_i = (\theta_{i,1}^{(k)},\ldots,\theta_{i,d}^{(k)})^\top,$ where $\theta_{i,j}^{(k)}$ is defined in Assumption \ref{alt-as}. Note that $\mathbb{E}[{\bf \Theta}^{(k)}_0] = {\bm 0}$.  Let
$$
{\bf \Gamma}^{(k)}_{\ell,\theta}=
\mathbb{E}\big[{\bf \Theta}^{(k)}_0{\bf \Theta}^{(k)\top}_{\ell}\big]
$$
and ${\bf \tilde{\Theta}}^*_i = (\tilde{\theta}_{i,1}^*,\ldots,\tilde{\theta}_{i,d}^{*})^\top$, where
$$
\tilde{\theta}_{i,j}^*=
\langle X_i\otimes X_i-\tau C_1-(1-\tau)C_2,\varphi_j^*\otimes\varphi_j^*\rangle.
$$
Define
$$
\tilde{{\bf \Gamma}}^*_{\ell,\theta}
=\frac{1}{n}\sum_{i\in\mathcal{I}_\ell}\big[{\bf \tilde{\Theta}}^*_i-\bar{{\bf \Theta}}^*\big]
\big[{\bf \tilde{\Theta}}^*_{i+\ell}-\bar{{\bf \Theta}}^*\big]^\top,
\qquad\mbox{and}\qquad
\tilde{\Sigma}_{\star,d} = \sum_{\ell= -\infty}^{\infty} w\left(\frac{\ell}{h}\right) \tilde{{\bf \Gamma}}_{\ell,\theta}^*.
$$
Using \eqref{var-phi-alt} and \eqref{c-alt-approx}, one may show as in \eqref{sig-cons-0} and \eqref{sig-cons-1} that
\begin{align}\label{sig-til-approx}
|\hat{\Sigma}_d - \tilde{\Sigma}_{\star,d} |_F = O_P\bigg(\frac{h}{\sqrt{n}}\bigg).
\end{align}
Adding and subtracting $C_k(t,t^\prime)$ in the integrand defining $\tilde{\theta}_{i,j}^*$, if follows that
\[ \tilde{\theta}_{i,j}^* = \begin{cases}
\langle X_i\otimes X_i-C_1,\varphi_j^*\otimes\varphi_j^*\rangle
+(1-\tau)\langle C_1-C_2,\varphi_j^*\otimes\varphi_j^*\rangle, \qquad i\leq k^*. \\
\langle X_i\otimes X_i-C_2,\varphi_j^*\otimes\varphi_j^*\rangle
+\tau\langle C_2-C_1,\varphi_j^*\otimes\varphi_j^*\rangle, \qquad\qquad\;\; i> k^*.
%
\end{cases}
\]
Therefore,
\[ {\bf \tilde{\Theta}}^*_i  = \begin{cases}
   {\bf \Theta}^{(1)}_i +  (1-\tau) ({\bm{\mu}}_1-{\bm{\mu}}_2),\qquad i \le k^*, \\
   {\bf \Theta}^{(2)}_i -\tau ({\bm{\mu}}_1-{\bm{\mu}}_2),\qquad\qquad\;\; i > k^*,
   \end{cases}
\]
where
$
{\bm{\mu}}_k = \langle C_k,\varphi_j^*\otimes\varphi_j^*\rangle
$.
Using the definition of $\bar{{\bf \Theta}}^*$ shows that
\begin{align*}
\bar{{\bf \Theta}}^*
&= \frac{\lfloor n \tau \rfloor }{n} \frac{1}{\lfloor n \tau \rfloor } \sum_{i=1}^{\lfloor n \tau \rfloor}  {\bf \tilde{\Theta}}^*_i + \frac{n - \lfloor n \tau \rfloor }{n} \frac{1}{n - \lfloor n \tau \rfloor } \sum_{i=\lfloor n \tau \rfloor+1}^{n}  {\bf \tilde{\Theta}}^*_i\\
&=  \frac{\lfloor n \tau \rfloor }{n} \frac{1}{\lfloor n \tau \rfloor } \sum_{i=1}^{\lfloor n \tau \rfloor}  {\bf \Theta}^{(1)}_i + \frac{n - \lfloor n \tau \rfloor }{n} \frac{1}{n - \lfloor n \tau \rfloor } \sum_{i=\lfloor n \tau \rfloor+1}^{n}  {\bf \Theta}^{(2)}_i \\
&\;\;\;\;\;\;\; +  \left(\frac{\lfloor n \tau \rfloor }{n}(1-\tau) -\frac{n - \lfloor n \tau \rfloor }{n}\tau \right)({\bm{\mu}}_1-{\bm{\mu}}_2).
\end{align*}
Therefore one may use Assumption \ref{alt-as} to show that
\begin{align}\label{theta-mean-conv}
\|\bar{{\bf \Theta}}^*\|_d = O_P\bigg(\frac{1}{\sqrt{n}}\bigg).
\end{align}
For $\ell \ge 1$, write
\begin{align*}
\tilde{{\bf \Gamma}}^*_{\ell,\theta} &= \frac{1}{n} \sum_{i=1}^{n-\ell }\big[\tilde{{\bf \Theta}}^*_i-\bar{{\bf \Theta}}^*\big]\big[\tilde{{\bf \Theta}}^*_{i+\ell}-\bar{{\bf \Theta}}^*\big]^\top  \\
&=\tau \frac{1}{n\tau} \sum_{i=1}^{\lfloor n \tau \rfloor -\ell }\big[\tilde{{\bf \Theta}}^*_i-\bar{{\bf \Theta}}^*\big]\big[\tilde{{\bf \Theta}}^*_{i+\ell}-\bar{{\bf \Theta}}^*\big]^\top \\
&\;\;\;\;\; + (1-\tau)\frac{1}{n(1-\tau)} \sum_{i=\lfloor n\tau \rfloor - \ell }^{ n -\ell }\big[\tilde{{\bf \Theta}}^*_i-\bar{{\bf \Theta}}^*\big]\big[\tilde{{\bf \Theta}}^*_{i+\ell}-\bar{{\bf \Theta}}^*\big]^\top \\
&= \tau \frac{1}{n\tau} \sum_{i=1}^{\lfloor n \tau \rfloor -\ell }\big[\tilde{{\bf \Theta}}^*_i-{\bm{\mu}}_1+{\bm{\mu}}_1-\bar{{\bf \Theta}}^*\big]\big[\tilde{{\bf \Theta}}^*_{i+\ell}-{\bm{\mu}}_1+{\bm{\mu}}_1- \bar{{\bf \Theta}}^*\big]^\top \\
&\;\;\;\;\; + (1-\tau)\frac{1}{n(1-\tau)} \sum_{i=\lfloor n\tau \rfloor - \ell }^{ n -\ell }\big[\tilde{{\bf \Theta}}^*_i-{\bm{\mu}}_2+{\bm{\mu}}_2-\bar{{\bf \Theta}}^*\big]\big[\tilde{{\bf \Theta}}^*_{i+\ell}-{\bm{\mu}}_2+{\bm{\mu}}_2-\bar{{\bf \Theta}}^*\big]^\top.
\end{align*}
Expanding the last line, using \eqref{theta-mean-conv} and Assumption \ref{alt-as}, it follows that
$$
\Big|\tilde{{\bf \Gamma}}^*_{\ell,\theta}  - \big[\tau{\bf \Gamma}^{(1)}_{\ell,\theta}+ (1-\tau){\bf \Gamma}^{(2)}_{\ell,\theta} + \tau(1-\tau)({\bm{\mu}}_1-  {\bm{\mu}}_2)({\bm{\mu}}_1-  {\bm{\mu}}_2)^\top\big] \Big|_F = O_P\bigg(\frac{1}{\sqrt{n}}\bigg).
$$
Since
$\Sigma_d^{(k)} = \sum_{\ell = -\infty}^{\infty}{\bf \Gamma}^{(k)}_{\ell,\theta},$
one obtains as in the proof of Theorem \ref{sig-cons-thm} that
$$
\big|\tilde{\Sigma}_{\star,d} - \Sigma_{\star,n}\big|_F = O_P\bigg(\frac{h}{\sqrt{n}}\bigg),
$$
where
$$
\Sigma_{\star,n} = \tau \Sigma_d^{(k)} + (1-\tau) \Sigma_d^{(k)} + \tau(1-\tau)({\bm{\mu}}_1-  {\bm{\mu}}_2)({\bm{\mu}}_1-  {\bm{\mu}}_2)^\top  \sum_{\ell = -\infty}^{\infty} w\bigg(\frac{\ell}h\bigg).
$$
The bounded support of $w$ gives $\sum_{\ell = -\infty}^{\infty} w(\ell/h)=O(h)$.
By Assumption \ref{alt-as} (c), the matrix $\Sigma_{\star,n}$ is strictly positive definite for all $n$. If $\chi_j$, $j=1,\ldots, d$ are the ordered (decreasing) eigenvalues of $\Sigma_{\star,n}$, then there exist positive constants $k_{10}$ and $k_{11}$ so that $\chi_d > k_{10} >0$ and $\chi_1 \le k_{11} h$. The largest eigenvalue of $\Sigma_{\star,n}^{-1}$ is then bounded, and hence so is $|\Sigma_{\star,n}^{-1}|_F$.  According to equation (26) of Henderson and Searle (1981) and the sub-multiplicative property of the Frobenius norm, this implies
\begin{align*}
|\tilde{\Sigma}_{\star,d}^{-1} - \Sigma_{\star,n}^{-1} |_F &= |\Sigma_{\star,n}^{-1} ( \hat{\Sigma}_d - \Sigma_{\star,n})(I +  \Sigma_{\star,n}^{-1}(\hat{\Sigma}_d - \Sigma_{\star,n})^{-1} \Sigma_{\star,n}^{-1}  |_F \\
&\le  k_{12}|\Sigma_{\star,n}^{-1} |_F^2 ||\hat{\Sigma}_d - \Sigma_{\star,n}|_F \\
&\le  k_{13}  |\hat{\Sigma}_d - \Sigma_{\star,n}|_F \\
&= O_P\bigg(\frac{h}{\sqrt{n}}\bigg),
\end{align*}
where $I$ is the identity matrix in $\mathbb{R}^{d\times d}$. One obtains similarly that $|\tilde{\Sigma}_{\star,d}^{-1} - \hat{\Sigma}_{d}^{-1} |_F=O_P(h/\sqrt{n})$. Combining this with \eqref{kap-aprox} and the fact that the largest eigenvalue of $\Sigma_{\star,n}^{-1}$ is bounded, and the smallest eigenvalue of $\Sigma_{\star,n}^{-1}$ is of the order $O(1/h)$, for a positive constant $k_{14}$,
\begin{align*}
{\bf \kappa}_n(\tau)^\top \hat{\Sigma}_d^{-1} {\bf \kappa}_n(\tau) &= n\tau^2 (1-\tau)^2  ({\bf \Lambda}_d^{(1)}-  {\bf \Lambda}_d^{(2)})^\top \Sigma_{\star,n}^{-1} ({\bf \Lambda}_d^{(1)}-  {\bf \Lambda}_d^{(2)}) + O_P(1) \\
 &\ge k_{14} (n/h)  \| {\bf \Lambda}_d^{(1)}-  {\bf \Lambda}_d^{(2)}\|_d^2 + O_P(1) \\
 & \stackrel{P}{\to} \infty,
\end{align*}
 as $n\to\infty$. This now implies part (a) in conjunction with \eqref{j-approx}. Restricting to the $j$th eigenvalue gives part (b). Part (c) follows along similar lines, and so details are omitted.
\end{proof}

\end{document}